\documentclass[letterpaper, 10 pt]{IEEEtran}
\IEEEoverridecommandlockouts

\usepackage{amsmath,dsfont,amssymb}
\usepackage{amsthm}
\usepackage{graphicx}
\usepackage{amsfonts}
\usepackage{float}
\DeclareMathOperator{\epi}{\mathbf{epi}}
\DeclareMathOperator{\dom}{\mathbf{dom}}

\DeclareMathOperator{\Poly}{\mathbf{Poly}}
\DeclareMathOperator{\bfx}{\mathbf{x}}
\usepackage{algorithmicx}
\usepackage{algorithm}
\usepackage{algpseudocode}
\usepackage{mathrsfs}
\usepackage{lipsum}
\usepackage{subfig}
\usepackage{amssymb}
\usepackage{bbm}
\usepackage{float}
\usepackage{balance}
\usepackage{color}
\theoremstyle{plain}
\newtheorem{thm}{\textbf{Theorem}}
\newtheorem{lem}{\textbf{Lemma}}
\newtheorem{prop}{\textbf{Proposition}}

\newcommand\numberthis{\addtocounter{equation}{1}\tag{\theequation}}
\theoremstyle{definition}
\newtheorem{defn}{\textbf{Definition}}

\newtheorem{exmp}{\textbf{Example}}
\newtheorem{ass}{\textbf{Assumption}}
\theoremstyle{remark}
\newtheorem{rem}{\textbf{Remark}}

\allowdisplaybreaks
\usepackage{algpseudocode}
\algdef{SE}[DOWHILE]{Do}{DoWhile}{\algorithmicdo}[1]{\algorithmicwhile\ #1}%

\usepackage[normalem]{ulem}

\begin{document}

\title{\LARGE \bf
Safety-Critical Control Synthesis for Network Systems with Control Barrier Functions and Assume-Guarantee Contracts
}
\author{Yuxiao Chen, James Anderson, Karan Kalsi, Aaron D. Ames, and Steven H. Low
\thanks{Yuxiao Chen,  Aaron D. Ames, and Steven Low are with the Division of Engineering and Applied Science and Engineering, Caltech,
        Pasadena, CA, 91106, USA. Emails:
        {\tt\small \{chenyx,ames,slow\}@caltech.edu}}
      \thanks{James Anderson is with the Department of Electrical Engineering and the Data Science Institute at Columbia University, New York, NY 10027, USA. {\tt\small james.andreson@columbia.edu}}
\thanks{Karan Kalsi is with Pacific Northwest National Laboratory, Richland, WA, 99352, USA. Email:
        {\tt\small Karanjit.Kalsi@pnnl.gov}}
\thanks{This work is supported by the Battelle Memorial Institute, Pacific Northwest Division, Grant \#424858.}
}

\maketitle
\begin{abstract}
This paper aims at the safety-critical control synthesis of network systems such that the satisfaction of the safety constraints can be guaranteed. To handle the large state dimension of such systems, an assume-guarantee contract is used to break the large synthesis problem into smaller subproblems. Parameterized signal temporal logic (pSTL) is used to formally describe the behaviors of the subsystems, which we use as the template for the contract. We show that robust control invariant sets (RCIs) for the subsystems can be composed to form a robust control invariant set for the whole network system under a valid assume-guarantee contract. An epigraph algorithm is proposed to solve for a contract that is valid, ---an approach that has linear complexity for sparse networks, which leads to a robust control invariant set for the whole network system. Implemented with control barrier function (CBF), the state of each subsystem is guaranteed to stay within the safe set. Furthermore, we propose a contingency tube Model Predictive Control approach based on the RCI, which is capable of handling severe contingencies, including topology changes of the network. A power grid example is used to demonstrate the proposed method. The simulation result includes both set point control and contingency recovery, and the safety constraint is always satisfied.
\end{abstract}
\vspace{-0.4cm}
\section{Introduction}\label{sec:intro}
Network control systems are sometimes subject to safety constraints that should be satisfied the entire time a system is a running. In the event that the network experiences a ``sudden or dramatic change'', either through subsystem failure, malicious attack, or unforeseen disturbance, it is imperative that the perturbed system not only maintains stability, but also still satisfies its safety constraints.

An example of such a network is the power grid. It is well known that if not controlled properly, cascading failures may lead to large-scale blackouts. The consequences of which can have a huge impact on infrastructure and economy, and in  the  worst case, lead to loss of life.

Traditional control techniques usually cannot guarantee the satisfaction of such constraints. One promising solution is ``correct-by-construction synthesis'', which has seen recent success in safety-critical applications such as vehicle control \cite{nilsson2014preliminary,chen2018validating} and robot navigation \cite{chen2018obstacle}. Correct-by-construction synthesis refers to a collection of methods (including but not limited to; barrier functions, density functions, and model checking). They are based on concepts such as reachable sets and robust control invariant sets \cite{Bla99} that ensure controllers are capable of enforcing safety constraints. Informally, a \emph{robust control invariant set} $\mathcal{S}$ is a subset of the state space, such that given a dynamical system initiated from within $\mathcal{S}$, there exists a control policy such that the system state can be kept within $\mathcal{S}$ for all future time, in the presence of disturbances. Typically, correct-by-construction control synthesis relies on computational tools such as the Hamilton Jacobi PDE \cite{mitchell2005time}, Linear Matrix Inequalities (LMIs) \cite{khlebnikov2011optimization}, and sum-of-squares (SOS) programming \cite{prajna2004nonlinear}. Unfortunately,  these methods do not scale well with the state dimension of the system. This curse of dimensionality has limited the applications of  correct-by-construction control synthesis to systems with low state dimension. There has been efforts to break ``the curse of dimensionality,'' which, at the system level, typically utilize either compositional analysis~\cite{AndP11} or symmetry \cite{nilsson2016control}.

To the best of the authors' knowledge, the synthesis of robust invariant sets for network systems with heterogeneous subsystems and strong coupling between them remains an open problem. Power grids are prominent examples of systems that exhibit the problematic phenomena just described. Typically, they consist of various types of generation buses e.g., hydroelectric, solar, and wind plants, and load buses, all coupled via transmission lines and the need to balance supply and demand whilst attaining frequency synchronization.

The approach we propose to break the curse of dimensionality is to use assume-guarantee contracts \cite{alur1999reactive} to decompose the overall performance guarantee of the network into individual contracts that each subsystem in the network agree to. Every subsystem in the network can take the performance guarantee from other subsystems as assumptions, and in turn provide its own performance guarantee, which then becomes part of the assumptions for other subsystems in the network. In this way, the big synthesis problem is decomposed into small subproblems. One related work is \cite{shamma2006decomposition}, where the authors use decentralized Lyapunov functions to quantify the coupling between subsystems, yet the computation of the Lyapunov functions was not discussed for general nonlinear systems.

The contributions of this paper are:

\noindent(i) We propose the formulation of an assume-guarantee contract approach to compute \emph{robust control invariant sets} (RCIs) for networked systems by combining subsystem RCIs with a network assume-guarantee contract.

\noindent(ii) We propose an epigraph algorithm that searches for valid assume-guarantee contracts, which has a computational complexity that scales linearly with system size (assuming the system graph is sparse or the coupling signals from multiple neighbors are summable). Moreover, the epigraph algorithm is general-purpose and can be combined with any RCI computation method to compute RCIs for network systems.

\noindent(iii) We propose a contingency tube MPC algorithm based on assume-guarantee contracts for set invariance, which is real-time implementable and is able to handle severe contingencies such as a change in the network topology.


\noindent \textit{\textbf{Nomenclature:}} $\mathbb{B}$, $\mathbb{N}$, and $\mathbb{R}$ denote the sets of binary variables,  natural numbers, and  real numbers, respectively. The $n$-dimensional Euclidean space is denoted by $\mathbb{R}^n$  and $\mathbb{R}^n_+$ denotes the non-negative orthant. Bold characters denote continuous or discrete-time signals, depending on the context, i.e.,  $\bfx= \{x(t)\}_{t=0}^{\infty}\in \mathcal{X}^\mathbb{N}$ when it is a discrete-time signal; $\bfx= \{x(t)\}_{t\in [0,\infty)}\in \mathcal{X}^\mathbb{R_+}$ when it is a continuous-time signal. $x(t)\in\mathcal{X}$ is a vector denoting the value(s) of $\bfx$ at time $t$, $\mathcal{X}^\mathbb{N}$ and $\mathcal{X}^\mathbb{R_+}$ denote the space of discrete/continuous-time signals of $x$. Given the set $\mathcal X := \mathcal X_1 \times, \mathcal X_2 \times \hdots \times \mathcal X_n$, $x\downarrow\mathcal{X}_i$ denotes the projection of $x$ onto $\mathcal{X}_i$, i.e., $x\downarrow\mathcal{X}_i=x_i$ where $x_i\in\mathcal{X}_i$. To avoid confusion between temporal signals and value iterations, we use $p[i]$ to denote the value of a parameter $p$ after the $i^{\text{th}}$ iteration. $\Poly(P,q)=\left\{x\mid Px\le q\right\}$ denotes a polytope defined with matrix $P,q$.

\vspace{-0.3cm}
\section{Problem Setup}\label{sec:setup}
In this section, we present the problem setup and show how the power grid control synthesis can be handled with the proposed method.

\vspace{-0.2cm}
\subsection{Network System Dynamics}\label{sec:network_dyn}
We consider a network dynamic system consisting of subsystems with coupling dynamics. The couplings between neighboring subsystems are treated bounded disturbances. Therefore, the following product of subsystems is considered:\footnote{In general, a networked dynamical system model would be defined over a graph structure \cite{sandell1978survey}.However, because we view the coupling between systems as bounded disturbances, we can consider a network of dynamical systems as the more simple the product system.}

\begin{equation}\label{eq:network_system}
  \Sigma = \Sigma_1 \times \Sigma_2 \times ... \times \Sigma_N.
\end{equation}
It is assumed that  each subsystem can be written in the form
\begin{equation}\label{eq:dynamic_equation}
\Sigma_i := \left\{
 \begin{aligned}
x_i^+ &= {f_i}\left( {{x_i},{y_{\mathcal{N}_i}},{u_i},{d_i}} \right),\\
{y_i} &= {c_i}\left( {{x_i}} \right), 
\end{aligned}
\right.
\end{equation}
where $x_i\in\mathcal{X}_i\subseteq\mathbb{R}^{n_i}$ is the $i^{\text{th}}$ current state and  $x_i^+$ denotes the successor state. The control input is  $u_i\in\mathcal{U}_i\subseteq\mathbb{R}^{m_i}$, the exogenous disturbance is $d_i\in\mathcal{D}_i\subseteq\mathbb{R}^{l_i}$, and $y_{\mathcal{N}_i}$ denotes the vector of signals consisting of the outputs of all of the neighboring subsystems connected to subsystem $\Sigma_i$. The vector $y_{\mathcal{N}_i}$ can be further decomposed as
\begin{equation}\label{eq:NB}
  y_{\mathcal{N}_i} = \begin{bmatrix} y_{j_1} & \hdots & y_{j_{N_i}} \end{bmatrix}^\intercal, \forall~ j_1,\hdots,j_{N_i}\in \mathcal{N}_i,
\end{equation}
where $\mathcal{N}_i$ is the neighbor set of the $i^{\text{th}}$ node with cardinality $\left| {\mathcal{N}_i} \right| = {N_i}$. The full networked system dynamics is then:
\begin{equation}\label{eq:network_dyn_form}
\resizebox{1\hsize}{!}{$
\begin{aligned}
   x^+=f(x,u,d) &= \begin{bmatrix}
{{{f}_1}({x_1},{y_{{\mathcal{N}_1}}},{u_1},{d_1})}&
 \dots &
{{{f}_N}({x_N},{y_{{\mathcal{N}_N}}},{u_N},{d_N})}
\end{bmatrix}^\intercal,\\
  c(x) &= \begin{bmatrix} c_1(x_1) & \dots & c_N(x_N) \end{bmatrix}^\intercal.
\end{aligned}
$}
\end{equation}

The overall state space and output space are denoted as $\mathcal{X}=\mathcal{X}_1 \times ...\times \mathcal{X}_N$ and $\mathcal{Y}=\mathcal{Y}_1\times ...\times \mathcal{Y}_N$, respectively. Since the method was first proposed for fixed point control, it is assumed w.l.o.g. that the equilibrium point is at the origin, i.e. $f(0,0,0)=0$  and that $c(0)=0$.

Given the dynamics, the behavior of the $i^{\text{th}}$ subsystem is uniquely determined by $x_i(0)$, $\mathbf{y}_{\mathcal{N}_i}$, $\mathbf{u}_i$ and $\mathbf{d}_i$, let $\mathcal{I}_i=\mathcal{X}_i\times \mathcal{Y}_{\mathcal{N}_i}^\mathbb{N}\times \mathcal{U}_i^\mathbb{N}\times \mathcal{D}_i^\mathbb{N}$ denote the space of input signals and initial conditions of the system $\Sigma_i$ and $\mathcal{X}_i^\mathbb{N}$ is the space of all possible state signals of $\Sigma_i$. A dynamical system $\Sigma_i\subseteq2^{\mathcal{I}_i}\times 2^{\mathcal{X}_i^\mathbb{N}}$ is understood as a subset of possible input and state signal pairs.
\begin{rem}
The results in this paper can easily be extended to the case of continuous-time dynamical systems. However, the methods we use to compute robust control invariant sets are most naturally presented in discrete-time, hence our choice.
\end{rem}

\vspace{-0.2cm}

\section{Review of Major Tools}\label{sec:tool_review}
In this section, we review the major tools necessary to our approach, including control barrier functions and parameterized assume-guarantee contracts.

\vspace{-0.3cm}
\subsection{Control Barrier Functions}\label{sec:CBF}

The computed robust control invariant set will be enforced with a control barrier function (CBF). CBFs allow safety constraints which are enforced through barrier functions to be integrated  with performance objectives encoded through control Lyapunov functions. Given a set of allowable initial conditions $\mathcal X_0$, and an unsafe set $\mathcal{X}_d$, a CBF ensures that all trajectories of a dynamical system initiated from $\mathcal X_0$ never enter $\mathcal X_d$. Typically the computation of a CBF acquired through convex programming requires an existing stabilizing control law termed  the ``legacy controller''. The controller produced by the CBF is referred to as the ``supervisory controller''.

To accommodate for the discrete-time dynamics used in this paper, we adopt the result in \cite{ames2017control,agrawal2017discrete} and utilize a discrete-time zeroing control barrier function. Specifically, given a discrete-time dynamic system:
\begin{equation*}
x^+=f(x,u,d),~x\in\mathbb{R}^n,~u\in\mathcal{U},~d\in\mathcal{D},
\end{equation*}
a discrete-time CBF is a function $h:\mathbb{R}^n\to\mathbb{R}$ that satisfies
\begin{equation}\label{eq:discrete_CBF}
  \begin{aligned}
&\forall~ x \in {\mathcal{X}_0},&h(x) \ge 0\\
&\forall~ x \in {\mathcal{X}_d},&h(x) < 0\\
&\forall~ x \in \left\{ {x\mid h(x) \ge 0} \right\}, &\forall~ d \in \mathcal{D},\exists~ u \in \mathcal{U}\;\\
&~&\mathrm{s.t.}~ h(f(x,u,d)) \ge \gamma (h(x)),
\end{aligned}
\end{equation}
where $\gamma:\mathbb{R}\to\mathbb{R}$ satisfies $s^2\ge\gamma(s)\cdot s\ge 0$, i.e. $\gamma(s)$ has the same sign as $s$ and $|\gamma(s)|\le |s|$. The supervisory controller is then implemented with the CBF QP:
\begin{equation}\label{eq:discrete_CBF_QP}
  \begin{aligned}
u^\star = \mathop {\arg \min }\limits_{u \in \mathcal{U}} &\left\| {u - {u^0}(x)} \right\|^2\\
\mathrm{s.t.}~&h(f(x,u,d)) \ge \gamma (h(x)),
\end{aligned}
\end{equation}
where $u^0$ is the legacy controller's policy. The constraint set~\eqref{eq:discrete_CBF} is not always convex. We will show later that it is convex for the special case discussed in this paper.
\begin{rem}
  For the case when the disturbance set $\mathcal D$ is known, \eqref{eq:discrete_CBF} is realizable. When the disturbance set is unknown, it is straightforward to extend \eqref{eq:discrete_CBF} to a robust CBF which can be solved using quadratic programming, see \cite{nguyen2016optimal} as an example.
\end{rem}

It can be shown that under mild conditions, a CBF ($h(\cdot)$ in~\eqref{eq:discrete_CBF}) can be constructed with a properly chosen $\gamma(\cdot)$ ($\gamma(x)\equiv0$ is always a valid choice) from a robust control invariant set (RCI) that contains $\mathcal{X}_0$ and does not intersect  $\mathcal{X}_d$.

We use robust linear programming to compute a minimal robust control invariant set for each subsystem in the network~\cite{chen2018RCI}. The algorithm is described in Appendix \ref{sec:Robust_LP_review}. Note  that the contract-based framework we propose  and the epigraph algorithm introduced in Section~\ref{sec:epigraph}, are both compatible with \emph{any}  algorithm that constructs and RCI.

The robust linear programming algorithm generates a polytopic RCI $\Poly(P,q)$, where $P$ is a constant $m\times n$ matrix and $q\in\mathbb{R}^m_{>0}$. Note that the origin is always contained in the interior of the RCI. The CBF is defined as
\begin{equation}\label{eq:CBF}
  h\left( x \right) = \mathop {\min }\limits_k \frac{{{q_k} - {P_k}x}}{{{q_k}}},
\end{equation}
where $P_k$ is the $k^\text{th}$ row of $P$ and $q_k$ is the $k^\text{th}$ entry of $q$.
\vspace{-0.4cm}
\subsection{Parameterized Signal Temporal Logic}\label{sec:pSTL}
To break the ``curse of dimensionality'' for large network systems, we use assume-guarantee contracts to decompose the synthesis problem for the whole network into smaller subproblems for the subsystems \cite{alur1999reactive,puasuareanu2008learning}.  The language of the specifications is Signal Temporal Logic (STL), which is an extension of Linear Temporal Logic (LTL) that allows for real time and predicates over real-valued signals\cite{asarin2011parametric,raman2014model}. We note that LTL deals with discrete-time signals, whereas STL uses continuous-time signals. Since the dynamics we consider in this paper are in discrete-time, we extend an STL formula to discrete-time signals by considering sample instances, as discussed in \cite{fainekos2009robustness}. This is necessary since STL's ability to allow for parameterized propositions is needed.  A Signal Temporal Logic formula $\phi:\mathcal{X}^\mathbb{R_+}\to\mathbb{B}$ uses the following grammar:
\begin{equation*}
  \phi = \top \mid \mu \mid \lnot \phi \mid \phi_1\wedge\phi_2\mid\phi_1\mathbf{U}_I \phi_2,
\end{equation*}
where $\top$ is the logical tautology, $\mu:\mathcal{X}\to\mathbb{B}$  is the space of all continuous time signal of x. is a logic proposition, $\lnot$ is Boolean negation,  $\wedge$ is the Boolean \textbf{AND}. Finally,  the ``until'' operator $\mathbf{U}$ which takes $\phi_1$ and $\phi_2$ as arguments, is true if given the interval $I$, there exists a time $t\in I$ that $\phi_2$ is true, and before $t$, $\phi_1$ is always true. When $I$ is not specified, it is assumed that by default $I=[0,\infty)$. The validity of a formula $\mu$ with respect to the signal $\bfx$ at time $t$ can be determined as  \\

\begin{tabular}{lll}
$(\bfx,t)\models \mu$                             &iff & $x(t)$ satisfies $\mu$\\
$(\bfx,t)\models \lnot\phi$                      &iff & $x(t)\not\models \phi$\\
$(\bfx,t)\models \phi_1\wedge\phi_2$             &iff & $x(t)\models \phi_1$ and $x(t)\models \phi_2$\\
$(\bfx,t)\models \phi_1\mathbf{U}_{[a,b]}\phi_2$ &iff & {\begin{tabular}{l}
                                                                $\exists~ t'\in t+[a,b]$ s.t. $x(t)\models\phi_2$ \\
                                                                and $\forall~ t''\in[t,t'], x(t)\models\phi_1$
                                                              \end{tabular} }
                                                              \vspace{0.4cm}
\end{tabular}
where $\models$ and $\not\models$ stands for ``satisfies'' and ``does not satisfy'' respectively. A signal $\bfx\models\mu$ if $(\bfx,0)\models\mu$. From the above basic grammar, one can derive additional temporal operators such as $\lozenge_I \phi \doteq \top \mathbf{U}_I \phi$, meaning ``$\phi$ is eventually true during $I$,'' and $\square_I\phi \doteq \lnot(\lozenge_I \lnot\phi)$, meaning ``$\phi$ is always true in $I$''.

Given an STL formula $\phi$, $L(\phi) = \left\{\bfx\in\mathcal{X}^\mathbb{R_+}\mid \bfx \models \phi\right\}$ is the language of the formula. A partial order is defined among STL formulas as $\phi_1 \preceq \phi_2$ if $\forall~ \bfx\in \mathcal{X}^\mathbb{R_+}, (\bfx\models \phi_1) \Rightarrow (\bfx \models \phi_2)$, or equivalently, $L(\phi_1)\subseteq L(\phi_2)$.

A Parameterized Signal Temporal Logic (pSTL) formula is an STL formula with parameters. For example, $\phi = \square_{[a,b]}(x\ge c)$ can be represented as the following pSTL: $\varphi(a,b,c) = \square_{[a,b]}(x\ge c)$, where $a,b$ and $c$ are the parameters and $\varphi:\mathbb{R}^3\to(\mathcal{X}^\mathbb{R_+}\to\mathbb{B})$ is the pSTL template. For the rest of the paper, it is assumed that all the pSTL formulas are defined on partially ordered parameter domains. Given a parameter domain $\mathcal{P}$, the partial order is denoted as $\le_{\mathcal{P}}$. For a pSTL $\varphi$ with domain $\mathcal{P}_1$, if $\mathcal{P}_1$ is a subspace of $\mathcal{P}_2$, then $\forall~ p\in\mathcal{P}_2$, $\varphi(p) = \varphi(p_\downarrow\mathcal{P}_1)$, where $\downarrow$ denotes the projection onto $\mathcal{P}_1$.

\vspace{-0.3cm}
\subsection{Assume-Guarantee Contract for Network Systems}\label{sec:network_ag}
Finally, we present a framework that builds a large assume-guarantee contract from small subcontracts, which is then used for the RCI computation whole network. We adopt the definition of assume-guarantee contract from \cite{kim2017small}:
\begin{defn}[Assume-Guarantee Contract]
An assume-guarantee contract $\mathcal{C}$ for the dynamic system  $\Sigma$ is a pair of STL formulae $[\phi_a,\phi_g]$ consisting of an assumption $\phi_a$ and a guarantee $\phi_g$ that enforces the logical implication $\phi_a \to\phi_g$.
\end{defn}
An assume-guarantee contract $\mathcal{C}=[\phi_a,\phi_g]$ is true for a dynamic system $\Sigma$ if $\Sigma \cap L(\phi_a)\subseteq L(\phi_g)$, or written compactly as $\phi_a\wedge\Sigma\to\phi_g$ with a slight abuse of notation. Note that $\Sigma$ here is understood as a proposition, interpreted as ``a trace satisfies the system dynamics''.

\begin{defn}[Parameterized Assume-Guarantee Contract]
An assume-guarantee contract $\mathcal{C}=[\phi_a,\phi_g]$ is in parameterized form if there exists pSTLs $\phi_a = \varphi_a(p_a)$, $\phi_g = \varphi_g(p_g)$ and a mapping $\lambda:\mathcal{P}_a\to \mathcal{P}_g$ such that $\mathcal{C}(p_a)=[\varphi_a(p_a),\varphi_g(\lambda(p_a))]$.
\end{defn}
%
In particular, $\phi_a$ consists of two parts: ${\phi _a} =\phi_{ae}\wedge\phi_{af} = {\varphi _{ae}}\left( {{p_{ae}}} \right) \wedge {\varphi _{af}}\left( {{p_{af}}} \right)$, where $\phi _{ae}$ is the specification for exogenous environment behavior and $\phi _{af}$ is the feedback specification, which is understood as the specification that changes with other contracts.
\begin{defn}[Parameterized Network Assume-Guarantee Contract]\label{def:network_contract}
  For a network defined in \eqref{eq:network_system}, a parameterized network assume-guarantee contract consists of individual parameterized assume-guarantee contracts $\mathcal{C}_i$ for each subsystem $\Sigma_i$. Let $p_{ae}\in\mathcal{P}_{ae}$, $p_{af}\in\mathcal{P}_{af}$ and $p_{g}\in\mathcal{P}_{g}$ be the parameters for $\varphi_{ae}$, $\varphi_{af}$ and $\varphi_{g}$. Each subcontract $\mathcal{C}_i$ consists of $\phi_a^i=\varphi_{ae}^i(p_{ae}^i)\wedge\varphi_{af}^i(p_{af}^i)$ and $\phi_g^i=\varphi_g^i(p_g^i)$. where $p_{ae}^i=p_{ae}\downarrow\mathcal{P}_{ae}^i$, $p_{af}^i=p_{af}\downarrow\mathcal{P}_{af}^i$ and $p_g^i=p_g\downarrow\mathcal{P}_g^i$. Then the network assume-guarantee contract is defined as $\mathcal{C}=[\phi_{ae}\wedge \phi_{af},\phi_g]$ with the parameter mapping $\Lambda:\mathcal{P}_{ae}\times \mathcal{P}_{af}\to\mathcal{P}_g$ and
  \begin{equation}
  \resizebox{.7\hsize}{!}{$
  \begin{aligned}
    \phi_{ae} &=&\varphi_{ae}(p_{ae}) &=&\bigwedge\limits_{i=1}^N{\phi_{ae}^i} &=&\bigwedge\limits_{i=1}^N{\varphi_{ae}^i}(p_{ae}^i), \\
    \phi_{af} &=&\varphi_{af}(p_{af}) &=&\bigwedge\limits_{i=1}^N{\phi_{af}^i} &=&\bigwedge\limits_{i=1}^N{\varphi_{af}^i}(p_{af}^i), \\
    \phi_g    &=&\varphi_g(p_g)       &=&\bigwedge\limits_{i=1}^N{\phi_g^i}    &=&\bigwedge\limits_{i=1}^N{\varphi_g^i}(p_g^i).
  \end{aligned}
  $}
  \end{equation}
\end{defn}


\section{Set Invariance with Assume-Guarantee Contracts}\label{sec:set_invariance}
We now present one of the main results of this paper, which utilizes a network assume-guarantee contract to prove set invariance for network systems.

\begin{thm}[Assume-guarantee reasoning]\label{thm:ag}
Consider the network system in \eqref{eq:dynamic_equation} associated with a parameterized network assume-guarantee contract defined in Definition \ref{def:network_contract} with parameter mapping $\Lambda$. Suppose the following are satisfied:

\noindent 1. Under the local mapping $\lambda_i:\mathcal{P}_{ae}^i\times\mathcal{P}_{af}^i\to\mathcal{P}_g^i$ for each subsystem, the following subcontract $\mathcal{C}_i:{\Sigma _i} \wedge \varphi _{ae}^i(p_{ae}^i)\wedge \varphi _{af}^i(p_{af}^i) \to \varphi _g^i(\lambda_i(p_{ae}^i,p_{af}^i))$ is satisfied for all $p_a^i \in \mathcal{P}_a^i\doteq \mathcal{P}_{ae}^i\times \mathcal{P}_{af}^i$, 

\noindent 2. There exists a mapping $\Gamma:\mathcal{P}_g\to\mathcal{P}_{af}$ such that $\varphi_g(p_g)\preceq \varphi_{af}^i(\gamma_i(p_g))$, where $\gamma_i(p_g)=\Gamma(p_g)\downarrow\mathcal{P}_{af}^i$.

\noindent 3. There exists $p_{ae}\in\mathcal{P}_{ae}$ such that $\varphi_{ae}(p_{ae})$ is true.

\noindent 4. There exists an initial feedback parameter $p_{af}[0]\in\mathcal{P}_{af}$ such that $\varphi_{af}(p_{af}[0])$ is true.

 Given $p_{ae}^i$, define $\hat{\lambda}_i(\cdot)=\lambda_i(p_{ae}^i,\cdot)$.
Let $\hat{\Lambda}(p_{af})=[\hat{\lambda}_1(p_{af}^1)^\intercal,\ \hat{\lambda}_2(p_{af}^2)^\intercal,\ ...\ \hat{\lambda}_N(p_{af}^N)^\intercal]^\intercal$,
then define recursively $p_g[k] = \hat{\Lambda}(p_{af}[k]),~  p_{af}[k+1]=\Gamma(p_g[k]).$
Under these conditions, the network system satisfies
\begin{equation}\label{eq:res_guarantee}
  \hat{\phi}_g=\bigwedge\limits_{k=0}^{\infty}{\varphi_g(p_g[k])}.
\end{equation}
\end{thm}
See the Appendix \ref{sec:proof_theorem_ag} for the proof.

\begin{figure}
  \centering
  \includegraphics[width=1\columnwidth]{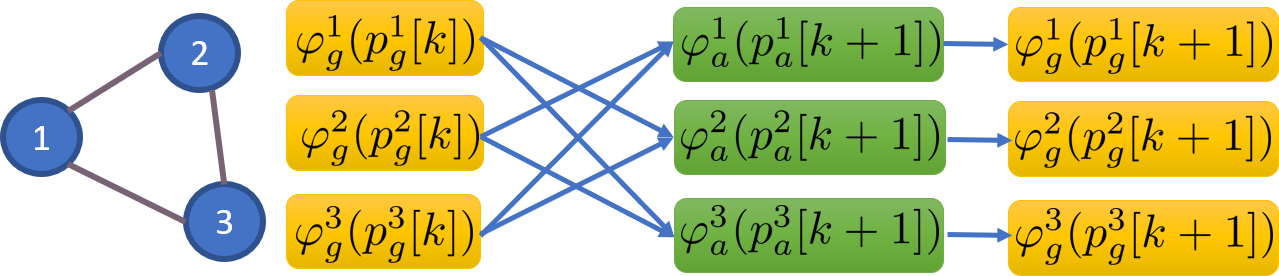}
  \caption{Network assume-guarantee contract for a 3-node network}\label{fig:3node}
  \vspace{-0.4cm}
\end{figure}

Fig. \ref{fig:3node} shows an example on a 3-node network. The guarantee of each node constitutes the assumption for the next iteration, and each node takes the assumption about its neighbors as an assumption from which a  guarantee is obtained.

Theorem \ref{thm:ag} can be viewed as the logical analogy of set invariance. If we have the recursive reasoning that propagates forward ($\phi_{af}$), and the initial logic proposition is satisfied ($\phi_{ae}$), then all the subsequent propositions are satisfied. Note that the guarantees on subsystems' behavior are shared across the network as assumptions for the next iteration.

Next, we apply Theorem \ref{thm:ag} to show set invariance of a network system. Consider a network system described by \eqref{eq:dynamic_equation}. Suppose that all subsystem outputs, $y_i$, are scalars\footnote{Scalar outputs are required for the epigraph algorithm (described on the next page) to work. Future work will relax this assumption.}, and for each subsystem $\Sigma_i$, $y_{\mathcal{N}_i}$ is treated as a disturbance. Then given a bound on $y_{\mathcal{N}_i}$: $\left|y_{\mathcal{N}_i}\right|\le y_{\mathcal{N}_i}^{\max}$, a bound $\mathcal{D}_i$ on $d_i$ and a bound $\mathcal{U}_i$ on $u_i$, any RCI algorithm can be applied to compute $\mathcal{S}_i$ for $\Sigma_i$ that satisfies
  \begin{equation*}
  \begin{array}{c}
    \forall~ x_i\in \mathcal{S}_i,\quad \forall~ d_i\in\mathcal{D}_i,\quad \forall~ \left|y_{\mathcal{N}_i}\right|\le y_{\mathcal{N}_i}^{\max}, \\
\exists~ u_i\in\mathcal{U}_i\; s.t. \quad    x_i^+={f_i}\left( {{x_i},{y_{\mathcal{N}_i}},{u_i},{d_i}} \right)\in \mathcal{S}_i.
  \end{array}
  \end{equation*}
Assume that $\mathcal{D}_i$ and $\mathcal{U}_i$ are given as part of the problem specification for all subsystems, the only information needed for the RCI computation is $y_{\mathcal{N}_i}^{\max}$. Let $\mathscr{F}$ be such a procedure that takes $y^{\max}$ as input, and computes an RCI. For clarity, we let $\mathscr{F}_i(y_{\mathcal{N}_i}^{\max})\subseteq \mathcal{X}_i$ be an RCI computed by $\mathscr{F}$ for the $i^{\text{th}}$ subsystem $\Sigma_i$, and let $\mathscr{F}(y^{\max})\doteq \mathscr{F}_1(y_{\mathcal{N}_1}^{\max})\times ...\times \mathscr{F}_N(y_{\mathcal{N}_N}^{\max})$ be the products of all the individual RCIs.

\begin{rem}
  Given a fixed procedure $\mathscr{F}$, it can be thought of as a mapping from the parameter $y^{\max}$ to the RCIs for the subsystems, which is then used to enforce constraints on the state. Note that $\mathscr{F}(y^{\max})$ is simply the product of RCIs for all the subsystems, and is not necessarily an RCI for the network system. It has to satisfy the validity condition defined later to be an RCI for the network system.
\end{rem}

\begin{defn}
  $\mathscr{F}$ is \textit{monotonic} w.r.t. $y^{\max}$ if given $y^{\max,1}\ge y^{\max,2}\ge {0}$, $\mathscr{F}(y^{\max,2})\subseteq \mathscr{F}(y^{\max,1})$. The inequality is defined element-wise. \end{defn}

\begin{lem}\label{lem:mon}
There always exists an $\mathscr{F}$ which is monotonic w.r.t. $y^{\max}$.
\end{lem}

\begin{proof}
$y^{\max,1}\ge y^{\max,2}$ implies that the uncertainty set for $\mathscr{F}(y^{\max,1})$ is a superset of the uncertainty set for $\mathscr{F}(y^{\max,2})$, so $\mathscr{F}(y^{\max,1})$ is also an RCI under $\left|y\right|\le y^{\max,2}$.
\end{proof}
\noindent The lemma above is intuitive since the size of the RCI should monotonically grow with the size of the disturbance bound. Under lemma~\ref{lem:mon}, we make the following assumption.
\begin{ass}\label{ass:F_monotone}
  The RCI computation procedure $\mathscr{F}$ considered in this paper is monotonic. Lemma~\ref{lem:mon} shows that this assumption can be made without loss of generality.
\end{ass}

Given a procedure $\mathscr{F}$ that computes RCIs for subsystems given $y^{\max}$ as described above, define the local mapping $\lambda_i$:
\begin{equation}\label{eq:RCI_lambda}
\begin{aligned}
  \lambda_i(y_{\mathcal{N}_i}^{\max})&\doteq\max\limits_{x_i\in\mathscr{F}_i(y_{\mathcal{N}_i}^{\max})}\left|h_i(x_i)\right|,\\
\Lambda(y^{\max})&\doteq[\lambda_1(y_{\mathcal{N}_1}^{\max});\lambda_2(y_{\mathcal{N}_2}^{\max});...;\lambda_N(y_{\mathcal{N}_N}^{\max})].
\end{aligned}
\end{equation}
Note that $\Lambda(y^{\max})$ has the same dimension as $y^{\max}$.
 Then we have our main theorem:
\begin{thm}[Set invariance of a network system with assume-guarantee contract]\label{thm:set_invariance}
  Given an RCI computation procedure $\mathscr{F}$ and let $\Lambda$ be defined in \eqref{eq:RCI_lambda}. If there exists a $y^{\max}\in\mathbb{R}^N_+$ such that
  \begin{equation}\label{eq:validity}
    \Lambda(y^{\max})\le y^{\max},
  \end{equation}
  then $\mathscr{F}(y^{\max})$ is an RCI for the network system.
\end{thm}
See the Appendix \ref{sec:proof_theorem_set_inv} for the proof.
%

The condition in \eqref{eq:validity} is the critical condition to show invariance, from hereon we refer to it as the ``\emph{validity condition}''. It can be interpreted as the condition  \emph{that each subsystem can satisfy what other nodes assume of it}. In the next section we will describe an algorithm that searches for a $y^{\max}$ that satisfies the validity condition when feasible.


\vspace{-0.3cm}
\section{Epigraph Method for Valid Contracts}\label{sec:epigraph}

In this section, we present an \emph{epigraph algorithm} that searches for an assume-guarantee contract that meets the validity condition~\eqref{eq:validity} if one exists. In particular, we show that the epigraph algorithm can be viewed as an extension of the classic \emph{small gain} theorem to network systems with nonlinear ``gains" and multiple interconnected systems.
\vspace{-0.4cm}
\subsection{Epigraph Representation of the Validity Condition}

Recall that given a function $g: \mathbb{R}^m \rightarrow \mathbb{R}$, the epigraph of $g$ is defined as
$$\epi(g) := \{(x,t)~|~ x\in \dom g,~ g(x)\le t \},$$ where $\dom g$ denotes the domain of $g$.

The idea behind our algorithm is to view each local $\lambda_i:\mathbb{R}^{N_i}_+\to\mathbb{R}_+$ as a function and consider its epigraph. Recall that $\lambda_i$ defined in \eqref{eq:RCI_lambda} denotes the mapping from the bounds on the outputs of the neighbors to the bound on the output of subsystem $i$. The condition in \eqref{eq:validity} is equivalent to the following condition:
\begin{equation*}
  [y_{\mathcal{N}_i}^{\max};y_i^{\max}]\in\epi(\lambda_i).
\end{equation*}
Note that this is a condition of only $y^{\max}$. Suppose the epigraph of each $\lambda_i$ is known, the search for a valid contract can be formulated as the following optimization:
\begin{equation}\label{eq:search_initial_contract}
  \begin{aligned}
\mathop {\min }\limits_{{y^{\max }} \ge {\bf{0}}} \;&\sum\nolimits_{i=1}^{N}{y^{\max}_i}\\
\mathrm{s.t.}~&\forall~ i = 1,...,N, \left[y_{\mathcal{N}_i}^{\max};y_i^{\max}\right]\in\epi(\lambda_i).
\end{aligned}
\end{equation}

\begin{exmp}
Consider the two systems $\Sigma_1$ and $\Sigma_2$ interconnected as shown in Fig. \ref{fig:two_interconnection}.
\begin{figure}[tb]
    \centering
    \includegraphics[width=0.6\columnwidth]{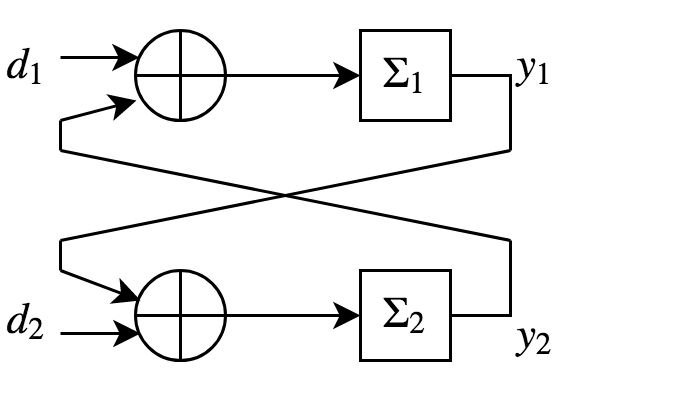}
    \caption{Two systems interconnection network}\label{fig:two_interconnection}
    \vspace{-0.4cm}
\end{figure}
\begin{figure}
  \centering
  \includegraphics[width=0.5\columnwidth]{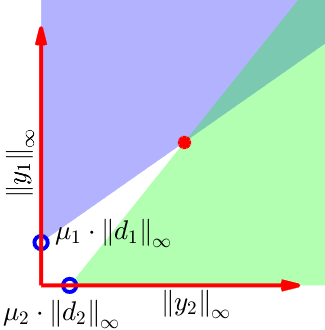}
  \caption{Epigraph view of the small gain theorem}\label{fig:epi_inter}
  \vspace{-0.5cm}
\end{figure}
Suppose that there exist constants $\mu_1,\mu_2,\nu_1, \nu_2 \ge 0$ such that
  \begin{equation}\label{eq:inter_assume}
    \begin{array}{l}
{\left\| {{y_1}} \right\|_\infty } \le {\mu _1}{\left\| {{d_1}} \right\|_\infty } + {\nu _1}{\left\| {{y_2}} \right\|_\infty },\\
{\left\| {{y_2}} \right\|_\infty } \le {\mu _2}{\left\| {{d_2}} \right\|_\infty } + {\nu _2}{\left\| {{y_1}} \right\|_\infty }.
\end{array}
  \end{equation}
If, in addition, the small gain condition is satisfied, i.e., $\nu_1\nu_2<1$, then the small gain theorem tells us that the interconnection is  stable and
\begin{equation}\label{eq:small_gain_res}
  \begin{aligned}
{\left\| {{y_1}} \right\|_\infty } &\le & \frac{{{\mu _1}}}{{1 - {\nu _1}{\nu _2}}}{\left\| {{d_1}} \right\|_\infty } &+ \frac{{{\mu _2}{\nu _1}}}{{1 - {\nu _1}{\nu _2}}}{\left\| {{d_2}} \right\|_\infty },\\
{\left\| {{y_2}} \right\|_\infty } &\le & \frac{{{\mu _1}{\nu _2}}}{{1 - {\nu _1}{\nu _2}}}{\left\| {{d_1}} \right\|_\infty } &+ \frac{{{\mu _2}}}{{1 - {\nu _1}{\nu _2}}}{\left\| {{d_2}} \right\|_\infty }.
\end{aligned}
\end{equation}
The proof can be found in \cite{desoer1975feedback}. The same result can be obtained by considering the epigraph.
\begin{prop}\label{prop:interconnected}
  Given \eqref{eq:inter_assume} and bounded $\left\| {d_i} \right\|_\infty >0 ,i=1,2$, there exists an assume-guarantee contract that guarantees \eqref{eq:small_gain_res} if $\nu_1\cdot\nu_2<1$.
\end{prop}
Due to space limitations, the  proof is omitted. It can be found in \cite{chen2018compositional}.

\end{exmp}
\begin{rem}
  This example shows that the small gain theorem can be viewed as a special case of the epigraph method, which is still applicable when $\lambda_i$ are nonlinear functions and when there are more than 2 interconnected subsystems.
\end{rem}

In practice, $\epi(\lambda_i)$ usually does not have a simple explicit form. Fortunately, we can replace $\epi(\lambda_i)$ in \eqref{eq:search_initial_contract} with a tractable inner approximation and the optimization would still generate a valid contract if feasible. The inner approximation of $\epi(\lambda_i)$ can be obtained by a grid sampling approach. To be specific; evaluate $\lambda_i$ at all the grid points (fixing $y_{\mathcal{N}_i}^{\max}$) by computing an RCI and evaluating \eqref{eq:RCI_lambda}. The approximation  of  $\epi(\lambda_i)$ is the area above the grid points in  $[y_{\mathcal{N}_i}^{\max};\lambda_i(y_{\mathcal{N}_i}^{\max})]$ space. The sampling complexity grows exponentially with the number of neighbors $N_i$. To reduce the sampling complexity, we introduce the notion  of \emph{summable signals}.
\begin{defn}
  Two or more disturbance signals are \textit{summable} if they have the same input dynamics. To be specific, consider $x^+=f(x,u,d)$, where $d=[d_1,...,d_l]^\intercal\in\mathbb{R}^l$ is the disturbance. The individual disturbances $\{d_i\}$ are summable if $\exists~ \bar{f}$ such that $f(x,u,d)\equiv\bar{f}(x,u,\sum_i d_i)$.
  \end{defn}
Summable disturbance inputs can be combined and viewed as one disturbance since they invoke the same disturbance dynamics and their bounds are summable, i.e.,
\begin{equation*}
  \left( {{{\left| {{d_1}} \right|}} \le \alpha } \right) \wedge \left( {{{\left| {{d_2}} \right|} } \le \beta } \right) \Rightarrow {\left| {{d_1} + {d_2}} \right| } = \left| d\right| \le \alpha  + \beta,
\end{equation*}
where the equality follows by definition.
Since the number of samples grows exponentially with the number of disturbance inputs, combining summable disturbance inputs reduces the complexity of the epigraph algorithm.


\vspace{-0.1cm}
\section{Power Grid Case Study}\label{sec:application}
We now apply our assume guarantee framework (and epigraph algorithm) to a power network case study.  We consider the problem of load-side primary frequency control \cite{mallada2017optimal}. The safety constraint considered is that the frequency deviation should never exceed a predefined bound. Frequency regulation is critical in power network. Modest deviations can damage electrical equipment and infrastructure (at the point of  load, generation and/or distribution), overload transmission lines leading  to market inefficiency, degrade the power quality delivered to consumers, and cause a network collapse if protective  systems kick in to protect equipment. In the US, the nominal frequency is $60$Hz and we further impose that the frequency deviation is below $0.05rad/s$. Broadly speaking, if power demand exactly matched supply, then frequency would not deviate from its set-point. However,  demand and supply can never be exactly matched; excess supply results in an increased frequency while a deficit causes frequency to decrease. Large deviations for extended periods of time will result in load shedding and potentially islanding.

 There has been a lot of effort focusing on the stability, optimality, and safety of power networks, see for example the survey paper \cite{molzahn2017survey}. Specifically, the Optimal Load Control (OLC) algorithms in \cite{zhao2014design,mallada2017optimal} provide control laws that can asymptotically track an optimal load-control problem i.e., control policy that achieves good asymptotic performance and maximizes economic benefit. To be more specific, the virtual flow method proposed in \cite{mallada2017optimal} formulates an OLC problem and derives a control policy based on a primal-dual update of the Lagrangian. However, despite good asymptotic performance, it lacks a performance guarantee in the transient phase. In particular, under contingencies such as large perturbation or topology change, the frequency deviation may exceed  the $\pm 0.05rad/s$ bound.

We shall use the OLC controller as the legacy controller to demonstrate the capability of the CBF controller proposed in Section~\ref{sec:CBF}. We will later show that robust control invariant sets with control barrier functions are a good complement to the OLC controller since it guarantees set invariance with minimum intervention and preserves the  performance of the OLC controller when the violation of safety constraints is not imminent.

\vspace{-0.4cm}
\subsection{Power Grid Dynamics}\label{sec:grid_dynamics}

We consider a transmission model of the power grid consisting of two types of buses; generators and loads. Take the IEEE 9-bus network depicted in Fig. \ref{fig:9bus_contingency} as an example. The set of generator buses are  $\mathcal G = \{1,2,3\}$, the remainder are pure load buses, the set of which is denoted by $\mathcal L$. The dynamics of the grid can be described by the following model \cite{mallada2017optimal}:
\begin{small}
\begin{align*}
  \dot{\theta}_i&=\omega_i,\numberthis \label{eq:grid_model}\\
  {M_i}{{\dot \omega }_i} &= P_i^{in} - {D_i}{\omega _i} - {r_i} - {u_i} - \sum\limits_{j \in {\mathcal{N}_i}} {\frac{{{V_i}{V_j}}}{{{X_{ij}}}}\sin } ({\theta _i} - {\theta _j}),i \in \mathcal{G}\\
  0&=P_i^{in}-D_i \omega_i - r_i-u_i -\sum\limits_{j\in\mathcal{N}_i}{\frac{{{V_i}{V_j}}}{{{X_{ij}}}}\sin } ({\theta _i} - {\theta _j}), i\in\mathcal{L},
\end{align*}
\end{small}

where $\theta_i$ and $\omega_i$ are the phase angle and frequency respectively of the voltage at bus $i$, $P_i^{in}$ and $r_i$ are the input power and uncontrollable load at bus $i$. A sudden change to either $P_i^{in}$ or $r_i$  is the main source of disturbance to the system, and the controllable load $u_i$ is used to regulate the power network. Each generator bus is modelled as a second order system with state $x_i=[\theta_i,\omega_i]^\intercal$, and $M_i$ and $D_i$ are the generator inertia and damping coefficient respectively; each load bus is modelled as a first order system with $x_i=\theta_i$ and zero inertia. $X_{ij}$ is the reactance of the circuit between bus $i$ and bus $j$. We choose the output to be $y_i=\theta_i$ since the coupling between buses occurs through the phase angles $\theta_i$. The model in \eqref{eq:grid_model} can be linearized and for each node, the subsystem dynamics $\Sigma_i$ are given by
\begin{align*}
 {\begin{bmatrix}
{\delta {{\dot \theta }_i}}\numberthis \label{eq:dyn_RCI}\\
{{{\dot \omega }_i}}
\end{bmatrix}} &=  {\begin{bmatrix}
0&1\\
{\frac{{ - \sum\limits_{j \in {\mathcal{N}_i}} {{B_{ij}}} }}{{{M_i}}}}&{\frac{{ - {D_i}}}{{{M_i}}}}
\end{bmatrix}}  {\begin{bmatrix}
{\delta {\theta _i}}\\
{{\omega _i}}
\end{bmatrix}}  +  {\begin{bmatrix}
0\\
-M_i^{-1}
\end{bmatrix}} {u_i}&\\
& +  {\begin{bmatrix}
0& \cdots &0\\
{\frac{{{B_{i{j_1}}}}}{{{M_i}}}}& \cdots &{\frac{{{B_{i{j_{{N_i}}}}}}}{{{M_i}}}}
\end{bmatrix}}  {\begin{bmatrix}
{\delta {\theta _{{j_1}}}}\\
 \vdots \\
{\delta {\theta _{{j_{{N_i}}}}}}
\end{bmatrix}} + e_i,&i \in \mathcal{G},\\
\\
\delta {{\dot \theta }_i} &= \frac{{ - \sum\limits_{j \in {N_i}} {{B_{ij}}} }}{{{D_i}}}\delta {\theta _i} + \frac{{ - 1}}{{{D_i}}}{u_i} + \frac{{\sum\limits_{j \in {\mathcal{N}_i}} {{B_{ij}}\delta {\theta _j}} }}{{{D_i}}} + {e_i},&i \in \mathcal{L}
\end{align*}
with output $y_i = \delta\theta_i$. $B_{ij}=\frac{{{V_i}{V_j}}}{{{X_{ij}}}}\cos(\theta_i^0-\theta_j^0)$ is the sensitivity of power flow to phase variations and $\theta_i^0$ is the steady-state phase angle at bus $i$. $B_{ij}\neq0$  when bus $i$ and $j$ are neighbors.

As mentioned before, the \textbf{control objective} is to prevent large frequency deviation from a set value. However, since the coupling is via the phase angle differences, in order to bound the frequency deviation, one needs to bound phase angle deviations as well. We will thus construct a robust control invariant set for both the phase angle and the frequency. The RCI should provide robustness to sudden changes of the input power $P_i^{in}$, uncontrollable load $r_i$ and the coupling between neighboring buses. We will treat the frequency deviation bound as the \textbf{safety constraint}, i.e., the danger set $\mathcal{X}_i^d$ for a generator bus $\Sigma_i$ is defined as
\begin{equation}\label{eq:X_d}
  \mathcal{X}^d_i=\left\{[\delta\theta_i,\omega_i]^\intercal\mid |\omega_i|\ge\omega^{\max}\right\},
\end{equation}
where $\omega^{\max}$ is the bound for frequency deviation. Note that the coupling between neighboring nodes happens via the phase angle, which is a scalar output. If we can use assume-guarantee contract to put bound on the phase angle deviations, we can compute an RCI for each node, which in turn constitute an RCI for the whole power grid network.

Let $\mathcal{X}_0$ be the set of allowed initial states, $\mathcal{X}_d$ be defined in \eqref{eq:X_d} for the generator buses (there is no $\mathcal{X}_d$ for pure load buses), we enforce additional constraint in the RCI computation such that the RCI does not contain $\mathcal{X}_d$, see Appendix \ref{sec:Robust_LP_review} for further details. With the polytopic RCI computed, the CBF is defined in \eqref{eq:CBF} and it can be verified that it satisfies \eqref{eq:discrete_CBF}. Then, the RCI can be enforced with the quadratic program~\eqref{eq:discrete_CBF_QP}.

\subsection{Epigraph Method for the search of valid Contracts}
To make sure that the CBF QP is always feasible, a robust control invariant set for the power network is needed. Since the goal is fixed point tracking, we use the linearized dynamics presented in \eqref{eq:dyn_RCI} for each node, and include the linearization error in the disturbance term. The assume-guarantee contract in this example follows the form introduced in Section \ref{sec:network_ag}. Each bus takes the bound on the phase angle deviation of its neighbors as the assumption, and guarantees that its own phase angle deviation stays bounded. The contract parameters are the bounds on phase angle deviation for each bus $\theta^{\max}$.

The computation of the RCI follows the robust linear programming algorithm \cite{chen2018RCI}. For each bus, the RCI is computed with the linearized model in \eqref{eq:dyn_RCI} after time discretization. The inputs to the RCI computation of the $i^{\text{th}}$ bus are the input sets $\mathcal{U}_i$ and exogenous disturbance bounds $\mathcal{D}_i$, given as the environment assumptions $\phi_{ae}^i$ (fixed), and bounds on the phase angle deviations of neighboring buses $\theta_{\mathcal{N}_i}^{\max}$, given as the feedback assumption $\phi_{af}^i$. Let $\mathscr{F}$ be the RCI computation procedure, and define
\begin{equation}\label{eq:lambda_grid}
  \lambda_i(\theta^{\max}_{\mathcal{N}_i})=\max\limits_{x_i\in\mathscr{F}(\theta^{\max}_{\mathcal{N}_i})}{\left|\theta_i\right|}.
\end{equation}
In the IEEE 9 bus example (as shown in Fig.~\ref{fig:9bus_contingency}), we add an additional constraint to $\mathscr{F}$ such that for each RCI, $\mathcal{S}_i$, computed for the generator buses,
$
  \max\limits_{x_i\in\mathcal{S}_i} \left|\omega_i\right|\le\omega^{\max},
$
so that $x_i\in\mathcal{S}_i$ implies that the safety constraint is satisfied.

By Assumption \ref{ass:F_monotone}, $\lambda_i$ is clearly monotonic. The evaluation of $\lambda_i$ is done in two steps. First, with $\theta_{\mathcal{N}_i}^{\max}$ fixed, $\mathscr{F}$ is called to compute an RCI $\mathcal{S}_i$, then $\theta_i^{\max}$ is obtained through \eqref{eq:lambda_grid}. Next, the inner approximation of $\epi(\lambda_i)$ is computed for each bus with the grid sampling algorithm, shown in Fig. \ref{fig:epi_grid_example}.
\begin{figure}[tb]
  \centering
  \includegraphics[width=0.85\columnwidth]{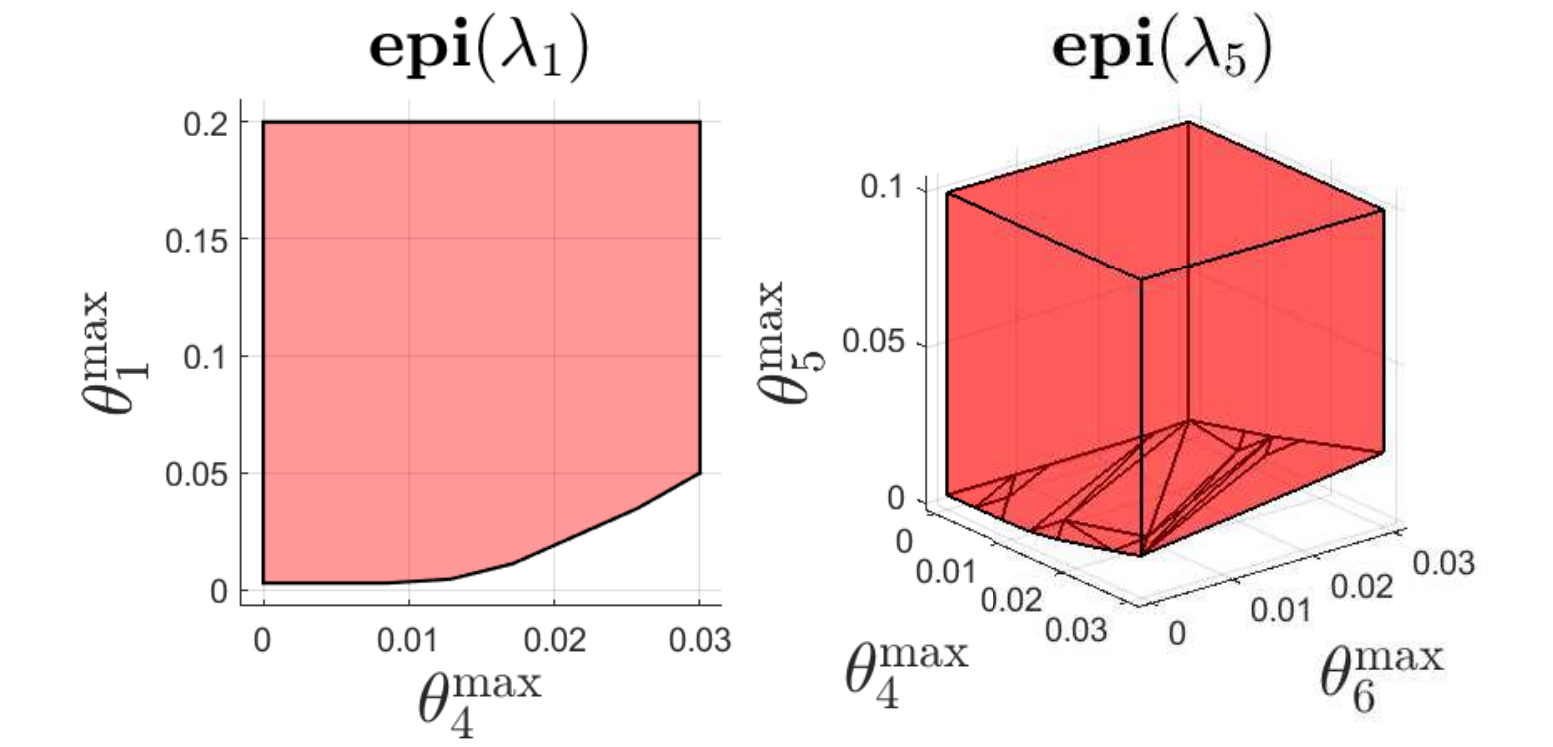}
  \caption{Inner approximations of $\epi(\lambda_1)$ and $\epi(\lambda_5)$}\label{fig:epi_grid_example}
   \vspace{-0.3cm}
\end{figure}
\begin{rem}
Clearly, the power flow from neighboring buses are summable, we treat them as separate disturbances in Fig. \ref{fig:epi_grid_example} to visualize the epigraph method with multiple non-summable disturbance inputs.
\end{rem}

As shown in Fig. \ref{fig:9bus_contingency}, bus 1 has one neighbor (bus 4) and bus 5 has two neighbors (bus 4 and 6), therefore $\epi(\lambda_1)$ is 2-dimensional whereas $\epi(\lambda_5)$ is 3-dimensional.
Once a value for $\theta^{\max}$ that satisfies the validity condition is found, it leads to a valid network assume-guarantee contract, and an RCI can be obtained via $\mathscr{F}$.


 Fig. \ref{fig:RCI} shows the robust invariant sets for the generator buses under the assume-guarantee contract, which satisfies the \textbf{safety constraint} that $|\omega_i|\le\omega^{\max}=0.05rad/s$.

\begin{figure}[tb]
  \centering
  \includegraphics[width=1\columnwidth]{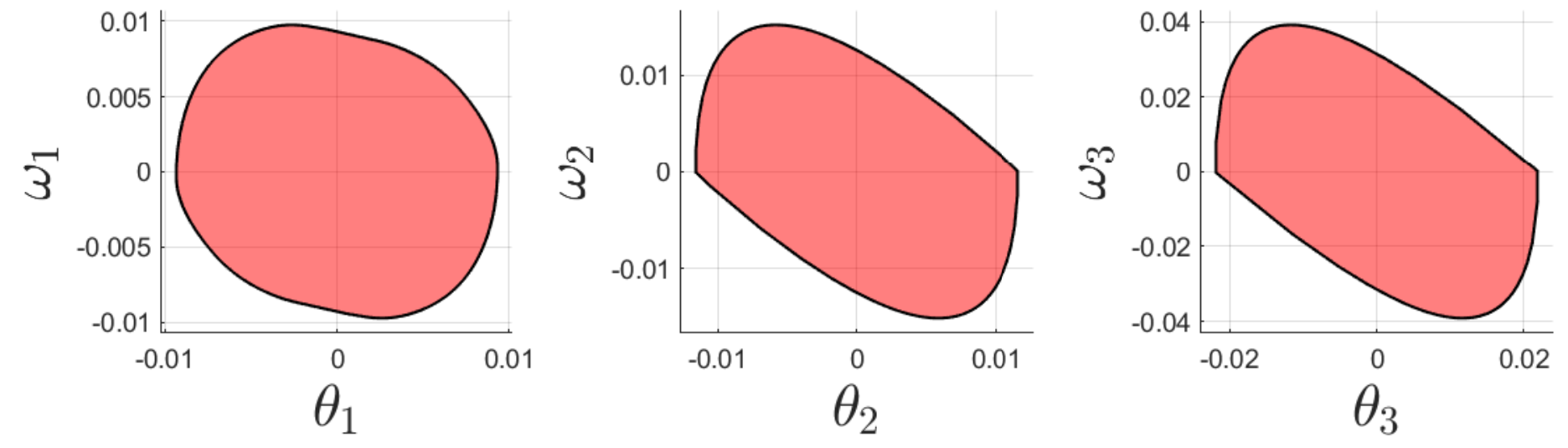}
  \caption{Robust control invariant sets for the generator buses }\label{fig:RCI}
  \vspace{-0.5cm}
\end{figure}

\vspace{-0.3cm}
\subsection{Simulation Result}\label{sec:9b_sim}

For each bus, the computed robust control invariant set is then enforced with a control barrier function as described in Section \ref{sec:CBF} with the OLC controller introduced in \cite{mallada2017optimal} used as the legacy controller providing the policy $u^0$. In Fig \ref{fig:sim_CBF1} we show a simulation trace of the 9-bus system with the CBF controller as the supervisory controller, and the \textbf{safety constraint} with $\omega^{\max}=0.05rad/s$ is never breached.

\begin{figure}[tb]
\vspace{-0.4cm}
  \centering
  \includegraphics[width=0.8\columnwidth]{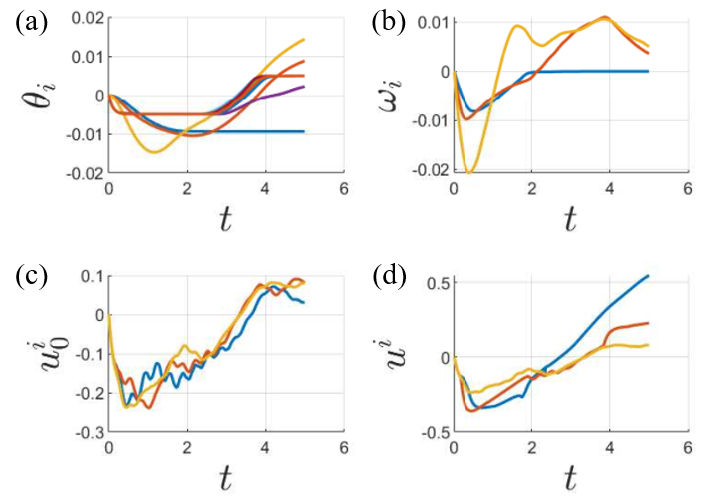}
  \caption{CBF control: Phase angle deviations of all 9 buses (a); generator frequency deviation (b); OLC legacy control (c); CBF supervisory control (d)} \label{fig:sim_CBF1}
  \vspace{-0.4cm}
\end{figure}


\begin{figure}[tb]
  \centering
  \includegraphics[width=0.8\columnwidth]{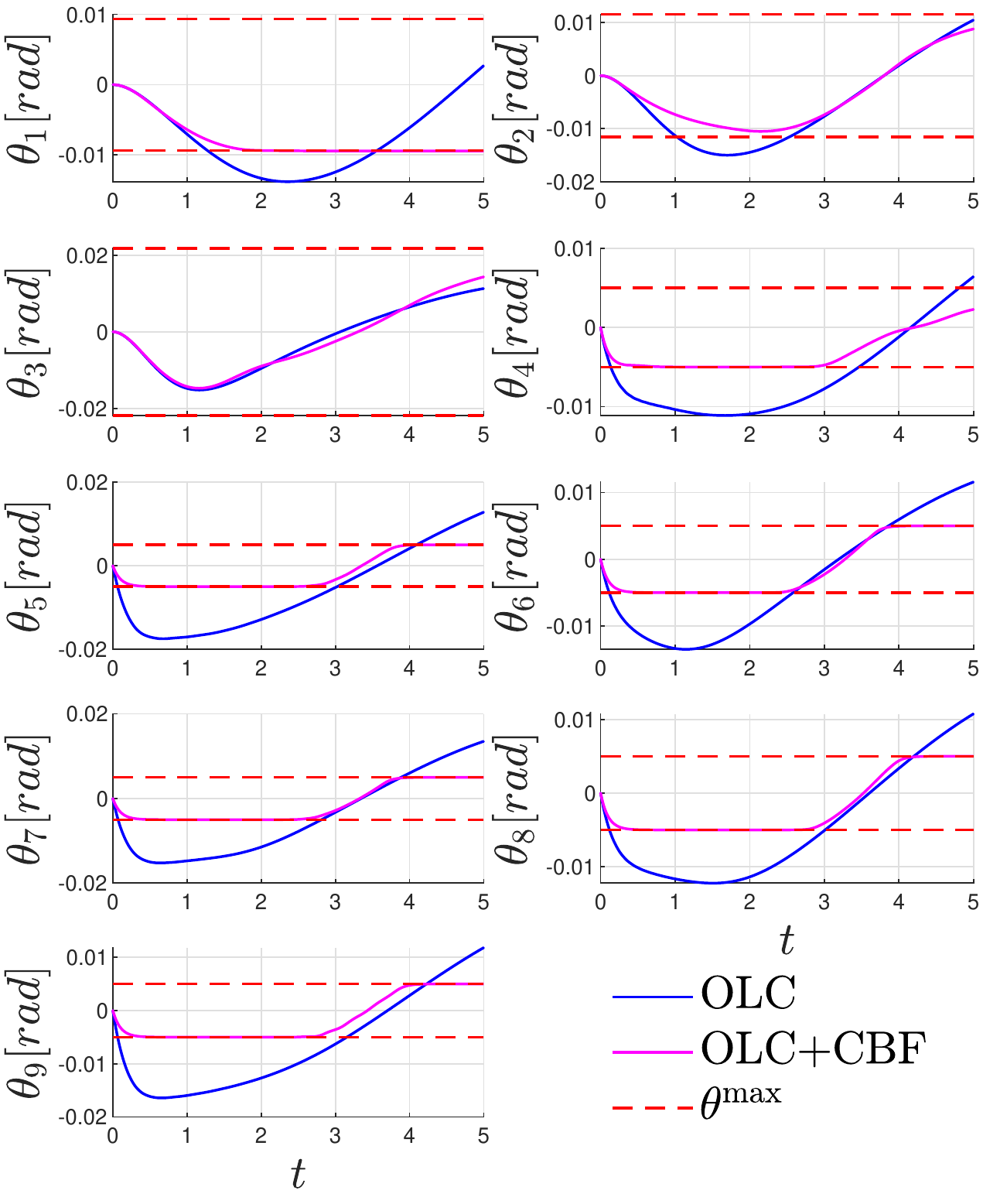}
  \caption{Phase angle plot with and without the CBF supervisor}\label{fig:sim_CBF2}
   \vspace{-0.5cm}
\end{figure}

Fig. \ref{fig:sim_CBF2} shows the phase angles with and without the CBF supervisor. Under the CBF supervisory controller (magenta plots), all phase anngles are within their respective bound determined by the contract; on the other hand, without CBF control (blue plots), there is no guarantee that the phase angles stay within bounds under the OLC policy $u^0$.

\vspace{-0.2cm}
\section{Model Predictive Control for Contingency Recovery}\label{sec:mpc}
We have shown how to compute an RCI for the network system with an assume-guarantee contract. We have shown this strategy is sufficient to guarantee the satisfaction of the safety constraint if the network operates around a fixed operating point $\theta^0$ (around which the dynamics are linearized).  However, when a severe contingency occurs such as a change in the network topology, a bus disconnects, or a line shorted, the RCI around the original operating point can no longer be maintained with the available control input, and the operating point has to change. This calls for an alternative controller that deals with the transient.

We propose a contingency tube model predictive controller that can handle the transient caused by contingency cases based on the mechanism developed for fixed point control.
\vspace{-0.4cm}
\subsection{Model Predictive Control for Reference Trajectory}\label{sec:MPC_ref}
Our MPC scheme is slightly different from the classic MPC (see for example~\cite{MPCbook16}), here we briefly review some concepts from the MPC literature and introduce our contingency tube MPC scheme.

There are two important horizons for MPC, the prediction horizon $T_p$ and the control horizon $T_c$. An MPC controller  looks ahead $T_p$ steps and represents the future state trajectory as a function of the input sequence, then solves for the optimal control sequence w.r.t. a cost function and some state and input constraints. The control sequence will be executed for $T_c$ steps, at which point another MPC iteration is executed, and a new control law computed. Traditional MPC schemes typically have $T_c\ll T_p$, often choosing $T_c=1$, which requires the controller to have access to the state information without delay.

In the network setting, distributed MPC schemes have been proposed  \cite{venkat2008distributed} that depend on fast communication and distributed optimization techniques. However, when the communication delay is not negligible,   the receding horizon scheme is likely to be infeasible. Instead, we consider a contingency tube MPC scheme that is triggered only when a contingency occurs, and the MPC nominal trajectory do not update until the end of the prediction horizon or another contingency occurs. Obviously, such an MPC scheme is equivalent to feedforward control once the MPC input is solved, and would not work without feedback. We use CBF at each node of the network as the feedback controller to guarantee the tracking performance of the reference trajectory generated by the MPC. The contingency tube MPC is designed to guarantee the safe transition of the network to the new operating point after the contingency.
\begin{figure}[tb]
  \centering
  \includegraphics[width=0.45\columnwidth]{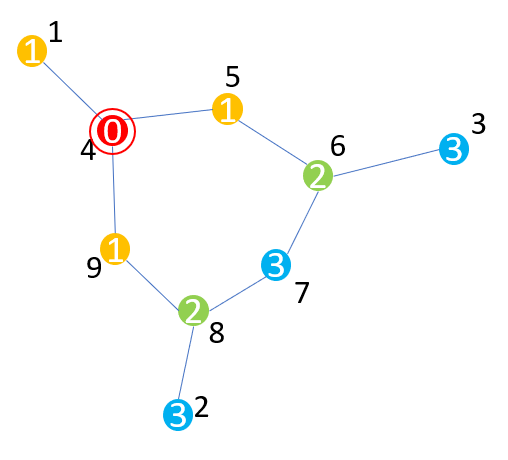}
  \caption{Contingency positions. Colors determined by delays: yellow (1-step delay), green (2-step delay), blue (3-step).}\label{fig:9bus_contingency}
  \vspace{-0.5cm}
\end{figure}
Three requirements for the MPC should be considered:
\begin{itemize}
  \item Computation of the MPC solution should be fast enough to allow real-time implementation.
  \item Safety constraints should be satisfied.
  \item Communication limitations should be respected.
\end{itemize}
Computation limitations differ with applications. In our power grid case study, in order to speed up the computation, we use the linearized model~\eqref{eq:dyn_RCI} and treat the nonlinearity  as a bounded disturbance. With a linear discrete-time model, quadratic costs and linear state and input constraints, the MPC can be solved by convex quadratic programming over the input sequence $\hat{u}(0:T_p-1)$. The MPC controller is triggered when any bus detects a contingency that exceeds the capability of the fixed point controller, such as connecting or disconnecting a bus or a line loss. To obtain the reference trajectory, the following optimization problem is solved:
\begin{equation}\label{eq:MPC_formulation}
\begin{aligned}
\min_{\hat u} \;&\mathcal{J}(\hat{u},\hat{x},x^{\star})  \\
\mathrm{s.t.}~ &{\hat{x}(t+1)} = \hat{f}(\hat{x}(t),\hat{u}(t),\hat{d}(t)),\\
&\forall~ i \in \mathcal{G}, t=0,1,...,T_p-1,\left| {{\omega _i}} \right| \le {\omega ^{\max ,ff}},\\
&\mathcal{C}\left( {{\hat{u}(0:T_p-1)}} \right) = 0,
\end{aligned}
\end{equation}
where $x^{\star}$ is the new operating point,  $\mathcal{J}$ is the cost function, which penalizes $\hat{u}$ and the distance between $\hat{x}$ and $x^\star$. $\omega ^{\max ,ff}$ is the bound on the bus frequencies for the MPC. Later we show that with CBF, the frequency tracking error is bounded by $\omega ^{\max ,fb}$. Let $\omega ^{\max}=\omega ^{\max ,ff}+\omega ^{\max ,fb}$, then the total frequency deviation is bounded by $\omega^{\max}$. $\hat{x}$ and $\hat{u}$ are the reference state and input trajectories and $\hat{d}$ is the predicted disturbance sequence, which depends on the knowledge of the contingency. In the general case $\hat{f}(x,u,d)$ is a linearization of \eqref{eq:network_dyn_form}, for this case study, the dynamics are given by  \eqref{eq:dyn_RCI}.
The set $\mathcal{C}$ is the constraint on the input caused by communication delay, which will be discussed later. The proposed scheme is based on the assumption that the network is close to a steady state when the contingency happens, therefore we can compute the reference trajectory for the whole network assuming that the system is at steady state without real-time state information. Once the MPC controller obtains a solution, the solution is sent to each node as the reference trajectory. Each node then uses a local feedback controller to track the reference trajectory.

Since the transmission of the reference trajectory is also subject to communication delay, we need the additional input constraint $\mathcal{C}$. Take the 9 bus test case as an example, suppose a contingency is detected at bus 4 and the MPC is computed at node 4. Assuming that the signal travels one edge per time-step, then the delay at each node is shown in Fig. \ref{fig:9bus_contingency}, and $\mathcal{C}$ would enforce the following input structure:
\begin{equation}\label{eq:input_structure}
\resizebox{.6\hsize}{!}{$
  \left[ {\begin{array}{*{20}{c}}
{\hat{u}_1(0:T_p-1)}\\
{\hat{u}_2(0:T_p-1)}\\
{\hat{u}_3(0:T_p-1)}\\
{\hat{u}_4(0:T_p-1)}\\
{\hat{u}_5(0:T_p-1)}\\
{\hat{u}_6(0:T_p-1)}\\
{\hat{u}_7(0:T_p-1)}\\
{\hat{u}_8(0:T_p-1)}\\
{\hat{u}_9(0:T_p-1)}
\end{array}} \right] = \left[ {\begin{array}{*{20}{c}}
0&*&*&*&*\\
0&0&0&*&*\\
0&0&0&*&*\\
*&*&*&*&*\\
0&*&*&*&*\\
0&0&*&*&*\\
0&0&0&*&*\\
0&0&*&*&*\\
0&*&*&*&*
\end{array}} \right],
$}
\end{equation}
which restricts the input to be zero before the reference trajectory signal arrives.

\vspace{-0.2cm}
\subsection{Contingency Tube MPC with CBFs}\label{sec:tube_MPC}
A local feedback controller is needed to track the reference trajectory generated by the MPC algorithm. The idea of centralized tube MPC was discussed in \cite{rakovic2011fully,yu2013tube}, and was extended to distributed tube MPC for multiple subsystems without coupling in the dynamics \cite{trodden2006robust}. There exist, however, strong coupling between nodes in the model of the grid dynamics~\eqref{eq:grid_model}. We use the assume-guarantee contract method proposed previously to handle the trajectory tracking problem for networks with strong coupling.

We assume that there exists a nominal dynamic model $\hat{f}_i$ for each subsystem in the network, and the difference between the model and the actual dynamics as described by~\eqref{eq:dynamic_equation} is bounded:
\begin{equation}\label{eq:f_hat}
  {f_i}\left( {{x_i},{y_{\mathcal{N}_i}},{u_i},{d_i}} \right)-\hat{f_i}\left( {{x_i},{y_{\mathcal{N}_i}},{u_i},{d_i}} \right)\in \mathcal{W}_{f_i},
\end{equation}
where  $\mathcal{W}_{f_i}$ is the bound for model mismatch.
The goal is to track a reference trajectory $\hat{x}(1:T_p)$ that satisfies
\begin{equation}\label{eq:ref_traj}
\begin{aligned}
  \hat x_i(t+1)  &= {{\hat f}_i}\left( {{{\hat x}_i}(t),{{\hat y}_{{\mathcal{N}_i}}}(t),{{\hat u}_i}(t),\hat{d}_i} \right),\\ \hat{y}_i(t)&=h(\hat{x}_i(t)),i=1,...,N,t=0,...,T_p-1
\end{aligned}
\end{equation}
 and keep the tracking error bounded. For the grid dynamics, $\hat{f}_i$ is linear, therefore we can write it as
 \[\hat{f}_i(x_i,y_{\mathcal{N}_i},u_i,d_i)=A_i x_i + B_i u_i + E^1_i y_{\mathcal{N}_i} + E^2_i d_i,\]
where $(A_i , B_i, E^1_i, E^2_i)$ are easily obtained from~\eqref{eq:dyn_RCI}.
Define the error $e_i=x_i-\hat{x}_i$, then the error evolves as
 \begin{equation*}
e_i^ +  = A_i e_i + B_i\Delta u_i+ E^1_i (y_{\mathcal{N}_i}-\hat{y}_{\mathcal{N}_i}) + E^2_i (d_i-\hat{d}_i) + \Delta f_i(t)
 \end{equation*}
where $\Delta u \doteq u_i-\hat{u}_i$ denotes the feedback part of the input, and $\Delta f_i(t)$ defined by
\begin{equation*}
{f_i}\left( {{{\hat x}_i}(t),{{\hat y}_{{\mathcal{N}_i}}}(t),{{\hat u}_i}(t),\hat{d}_i} \right) - {{\hat f}_i}\left( {{{\hat x}_i}(t),{{\hat y}_{{\mathcal{N}_i}}}(t),{{\hat u}_i}(t),\hat{d}_i} \right)
\end{equation*}
is the modeling error from linearization. We assume it is bounded and belongs to the set $\mathcal{W}_{f_i}$. For the  power grid case study, the error is caused by linearization of the sinusoidal functions. Since the reference trajectory has a finite duration $T_s T_p$ and $\omega_i$ is bounded by $\omega^{\max}$ for every bus, the bound on the modeling error can be obtained, c.f. Fig. \ref{fig:linearization_error}. We  now have everything in place to state the main result of this section.
\begin{figure}[tb]
  \centering
  \includegraphics[width=0.7\columnwidth]{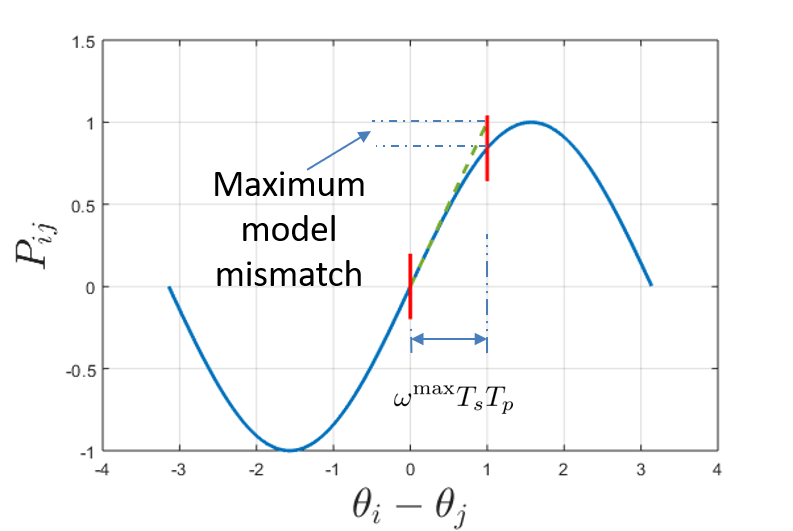}
  \caption{Linearization error}\label{fig:linearization_error}
  \vspace{-0.5cm}
\end{figure}

\begin{thm}
 Consider the power system dynamics in \eqref{eq:grid_model}, denoted as $f$, and the linearized model in \eqref{eq:dyn_RCI}, denoted as $\hat{f}$. For a reference trajectory $[\hat{x}(1:T_p),\hat{u}(0:T_p-1)]$ satisfying \eqref{eq:ref_traj}, if for each bus, there exists a feedback controller $\Delta u_i=k_i(x_i-\hat{x}_i, y_{\mathcal{N}_i}-\hat{y}_{\mathcal{N}_i},d_i-\hat{d}_i)$ such that for a given bound $\Delta \mathcal{D}_i$ of $d_i-\hat{d}_i$, a given set $\mathcal{S}_i\subseteq\mathcal{X}_i$ and a given bound on $\left|y-\hat{y}\right|\le \Delta y^{\max}$, the following is true:
  \begin{equation*}
  \begin{array}{c}
    \forall~ x_i(t)\in \hat{x}_i(t)+\mathcal{S}_i, d_i(t)\in\hat{d}_i(t)+\Delta\mathcal{D}_i\\ \forall~ \left|y_{\mathcal{N}_i}(t)-\hat{y}_{\mathcal{N}_i}(t)\right|\le \Delta y_{\mathcal{N}_i}^{\max}, \\
    x_i(t+1)={f_i}\left( {{x_i},{y_{\mathcal{N}_i}},{u_i},{d_i}} \right)\in \hat{x}_i(t+1)+\mathcal{S}_i\\
\max\limits_{x_i(t)\in\hat{x}_i(t)+\mathcal{S}_i}{\left|h_i(x_i)-\hat{y}_i\right|}\le \Delta y_i^{\max}, t=0,...,T_p-1,
  \end{array}
  \end{equation*}
  where $\Delta y_{\mathcal{N}_i}^{\max}$ is a projection of $\Delta y^{\max}$ onto $\mathcal{Y}_{\mathcal{N}_i}$ and the ``$+$'' signs between vectors and sets denote direct sums.
  Then let $u_i=\hat{u}_i+k_i(x_i-\hat{x}_i,y_{\mathcal{N}_i}-\hat{y}_{\mathcal{N}_i},d_i-\hat{d}_i)$, for any $x(0)$ satisfying $x_i-\hat{x}_i\in\mathcal{S}_i$, disturbance satisfying $d_i(t)\in\hat{d}_i(t)+\Delta \mathcal{D}_i$, the closed loop trajectory stays inside the tube defined as $\left\{x(1:T_p)~|~x(t)\in\hat{x}(t)+\mathcal{S}_1\times...\times\mathcal{S}_N\right\}$.
\end{thm}
\begin{proof}
  The proof can be obtained by directly applying Theorem \ref{thm:set_invariance} on the error dynamics.
\end{proof}

To implement the contingency tube MPC, we first compute an RCI for the error dynamics taking the bound on disturbance and model mismatch into account. When a contingency occurs, the MPC scheme  \eqref{eq:MPC_formulation} is solved to obtain a reference trajectory $\hat{x}$, then at each node, the following CBF supervisory control is implemented:
\begin{equation}\label{eq:MPC_CBF}
  \begin{aligned}
u_i(t)=\mathop {\arg\min }\limits_{u\in\mathcal{U}_i} \;&{\left\| {u- {u^0_i}(t)} \right\|^2}\\
s.t.\quad &\dot b_i(x_i-\hat{x}_i,u) + \kappa b_i(x_i-\hat{x}_i) \ge 0,
\end{aligned}
\end{equation}
where $b_i$ is the CBF for the $i^{\text{th}}$ node defined based on the RCI $\mathcal{S}_i$, $u^0_i$ is the nominal control signal for the $i^{\text{th}}$ node, which can be simply chosen as $\hat{u}_i$, or alternatively chosen as $\hat{u}_i$ plus a local feedback part. In the next section, the legacy control $u^0_i$ is picked as $\hat{u}_i$ plus an LQR feedback component.
\vspace{-0.2cm}
\subsection{Simulation of the Contingency Tube MPC}\label{sec:mpc_sim}
To validate the proposed contingency tube MPC scheme, we use the high-fidelity power grid simulator PST \cite{chow1992toolbox} as the simulation environment. PST allows several types of contingency cases, such as the 3-phase error, loss of line and loss of load. The New England network from PST is picked for demonstration, which contains 39 buses with 10 of them generator buses, as shown in Fig. \ref{fig:datane}. The red nodes are the generator buses and the green nodes are the pure load buses.
\begin{figure}[tb]
  \centering
  \includegraphics[width=0.60\columnwidth]{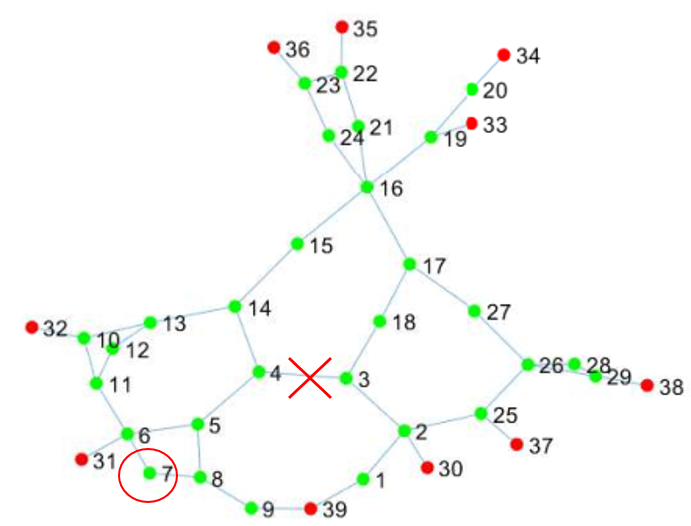}
  \caption{New England grid structure and failure locations}\label{fig:datane}
  \vspace{-0.5cm}
\end{figure}
The two tested contingencies are:
\begin{itemize}
\item \textbf{Case 1:} Load bus loss at bus 7
\item \textbf{Case 2:} Line between bus 3 and 4 trips
\end{itemize}
 When bus 7 disconnects, the network is able to find a new set point without changing the generation. When the line between bus 3 and 4 disconnects, the network cannot balance itself with the original generation. So an optimal power flow (OPF) routine (AC OPF routine in Matpower toolbox \cite{zimmerman1997matpower}) is performed to get the new generation together with the new operating point, and the contingency tube MPC is used to complete the transition to the new operating point. The sampling time and horizon for the contingency tube MPC is set at 50ms and 2.5s ($T_p=50$).

We insert sinusoidal load fluctuation with the maximum magnitude allowed by the RCI at every bus to simulate the effect of uncontrolled load disturbance. Once the contingencies (bus loss in case 1 and line loss in case 2) are detected, the contingency tube MPC kicks in at the nearest node to the contingency (bus 6 in case 1 and bus 4 in case 2) to compute the reference trajectory for the transition to the new operating points. Then the plan is then communicated across the network; the signal is assumed to travel two edges per sampling interval.
\begin{figure}[tb]
  \centering
  \includegraphics[width=0.8\columnwidth]{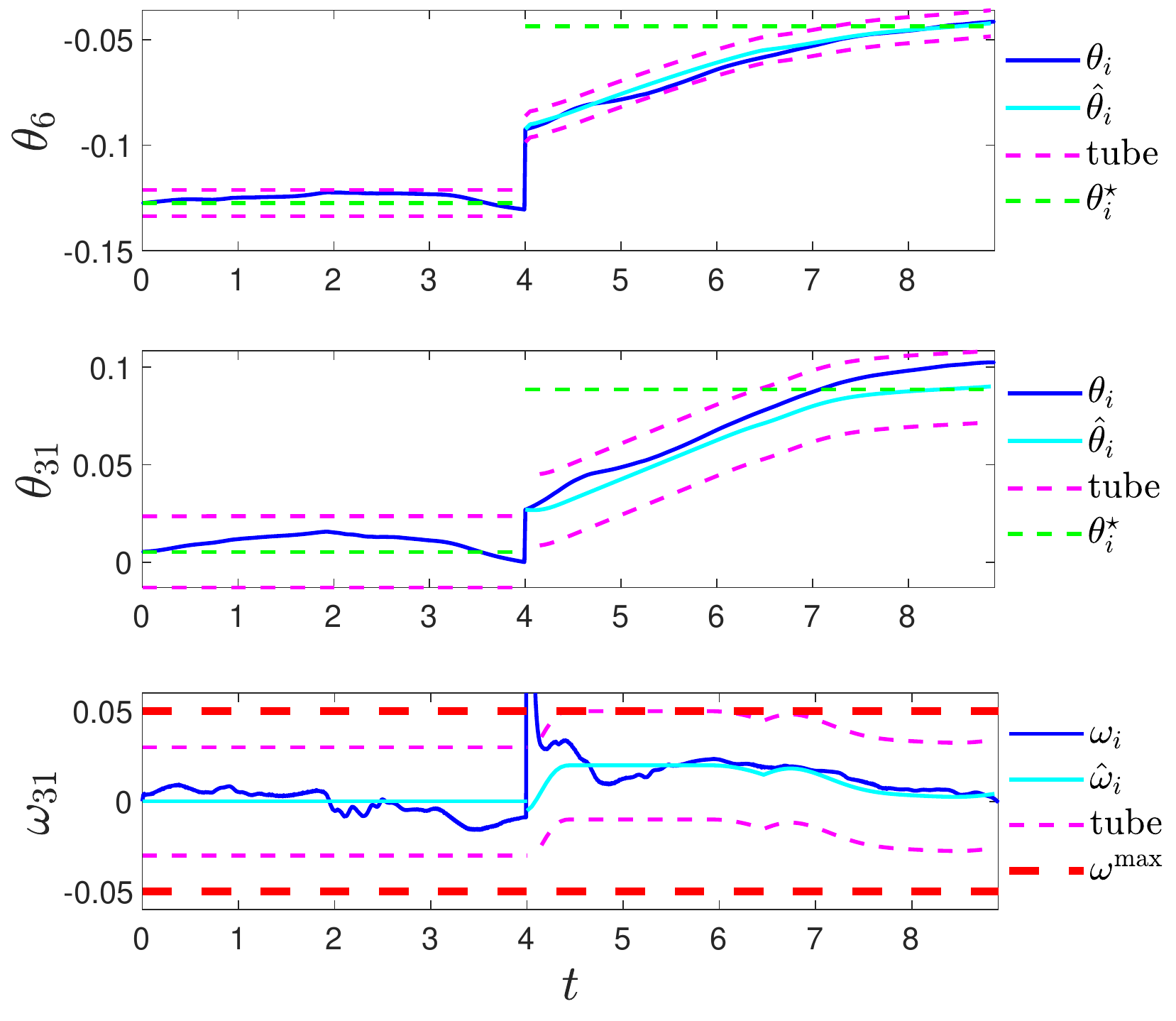}
  \caption{Case 1: Bus failure contingency}\label{fig:bus_loss_sim}
  \vspace{-0.6cm}
\end{figure}

Fig. \ref{fig:bus_loss_sim} shows the PST simulation of case 1, the load failure occurs at $t=4s$. The blue line is the state, the magenta line represents the tube (i.e. the region the state is confined to lie in), the green line represents the new set point for the phase angle and the red line represents the bound for frequency. We show the state trajectory of bus 6, the bus closest to the contingency, and bus 31, the closest generator bus to the contingency. When the contingency happens, the frequency breached the constraint for a slight moment, the reason for this violation of the safety constraint are (i) the dynamics under the contingency are not modeled accurately (ii) the violation happened instantly after the loss of bus 7, before the contingency tube MPC is able to kick in and react. Once the contingency tube MPC scheme kicks in, the state trajectory was kept within the tube and the network eventually reaches the new operating point without violating the safety constraint. In practice, infrequent small violations over very small time periods are tolerated.
\begin{figure}[tb]
  \centering
  \includegraphics[width=0.8\columnwidth]{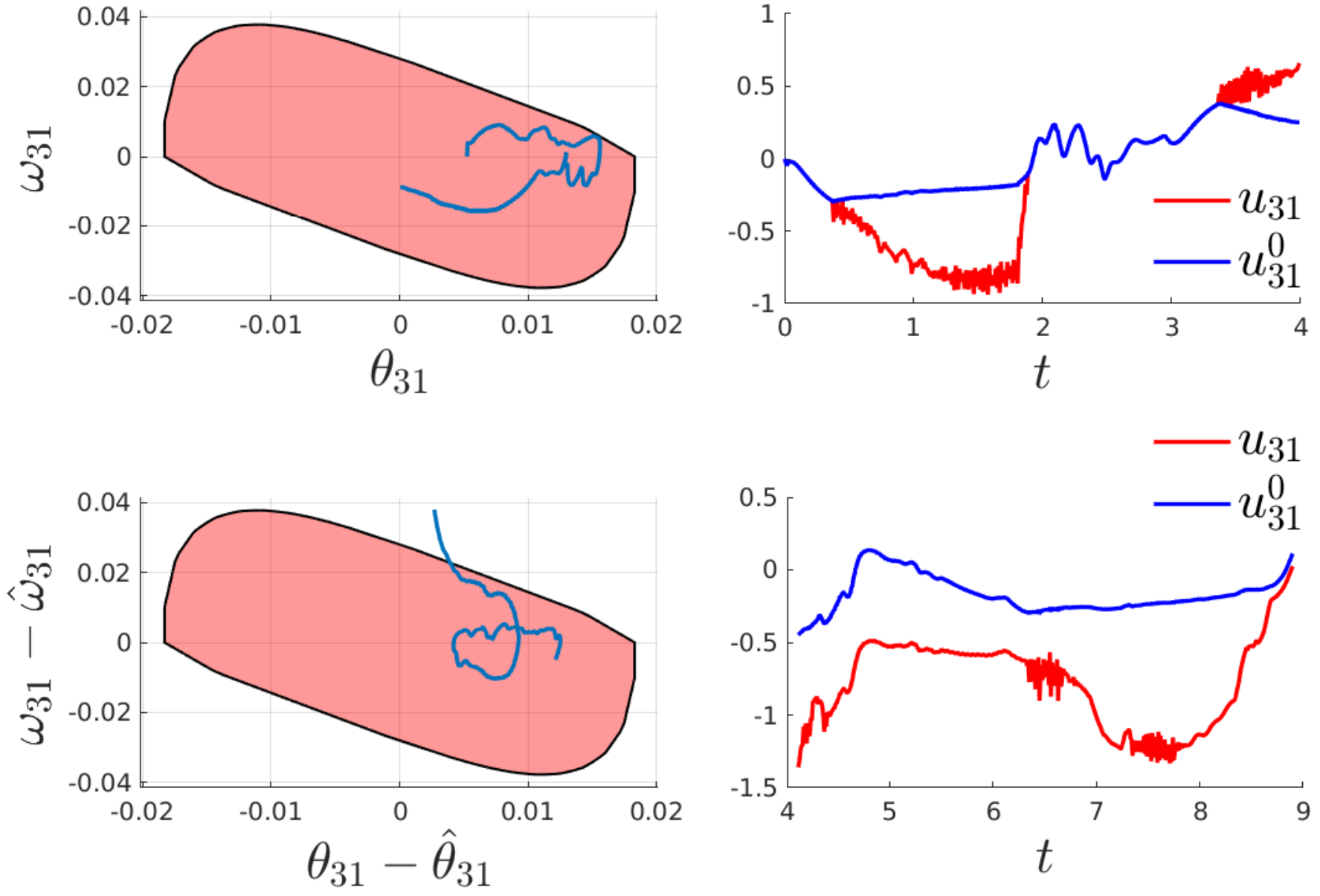}
  \caption{Robust Control Invariant set with state and input trajectories at bus 31}\label{fig:RCI_sim_plot}
  \vspace{-0.4cm}
\end{figure}
\begin{figure}[tb]
  \centering
  \includegraphics[width=0.8\columnwidth]{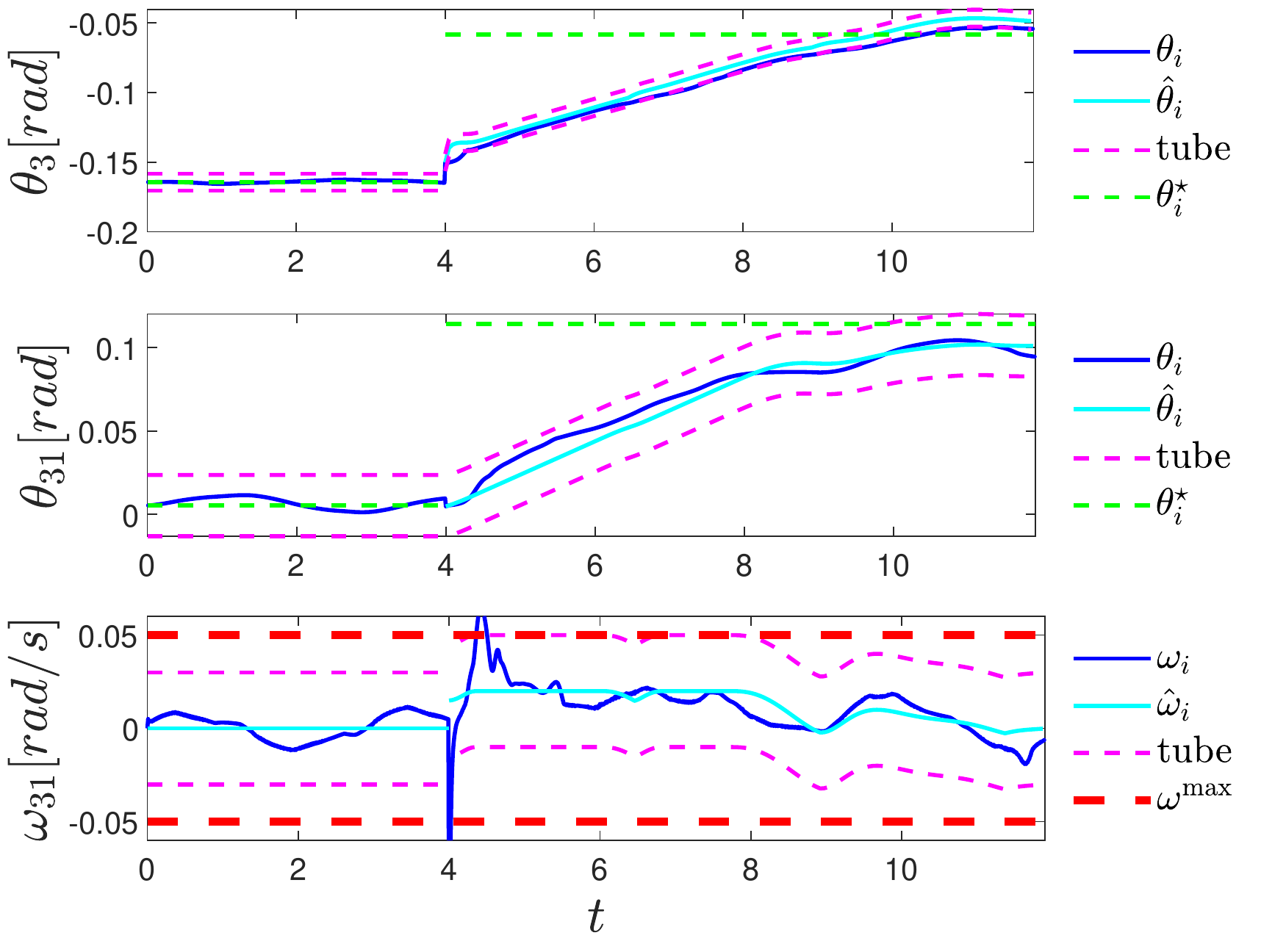}
  \caption{Case 2: Line failure contingency }\label{fig:line_loss_sim}
  \vspace{-0.5cm}
\end{figure}
In Fig. \ref{fig:RCI_sim_plot} we plot the state trajectory w.r.t. the RCI and the inputs to the system at bus 31 (generator bus). The two figures on top show the state and input trajectories before the contingency at $t=4s$. Due to the sinusoidal fluctuation of the load, the phase angle also fluctuates, but it never left the RCI; Fig. \ref{fig:RCI_sim_plot}(a). In Fig. \ref{fig:RCI_sim_plot}(b), the blue curve shows the legacy controller input, and the red curve shows the CBF controller input. The timing of the interventions coincide with the timing when the state is close to the boundary of the RCI. Fig. \ref{fig:RCI_sim_plot}(c) and (d) show the state and input trajectories after the contingency. Note that in the contingency tube MPC scheme, we require the error state $x-\hat{x}$ instead of the state $x$ to stay inside the RCI. After the temporary deviation right after the contingency, $x-\hat{x}$ stays inside the RCI due to the CBF controller.

Fig. \ref{fig:line_loss_sim} shows the simulation for the line loss case. Similarly, the contingency tube MPC together with the CBF controller is able to keep the system trajectory within the tube and take the whole network to the new operating point.

\vspace{-0.2cm}
\section{Conclusion}\label{sec:conclusion}
We consider the application of robust control invariant set and control barrier functions on network systems to prevent large deviations from the desired working condition. The key idea is to use assume-guarantee contracts to break the large network into small subsystems. The coupling between subsystems are treated as bounded disturbances which are handled with a network assume-guarantee contract. We show that a network assume-guarantee contract satisfying the validity condition guarantees robust set invariance for the whole network system. Furthermore, we propose an epigraph algorithm that searches for a valid contract, and enjoys linear complexity when the network is sparse or the coupling terms are summable. Based on the network assume-guarantee contract idea, we further propose a contingency tube MPC scheme that is capable of handling contingencies with changing operating points while respecting communication limitations. 

\appendices
\section{Robust Linear Programming for RCI Computation}\label{sec:Robust_LP_review}
We briefly review the robust linear programming algorithm for minimal robust control invariant set (mRCI) computation proposed in \cite{chen2018RCI}. A set is robust control invariant if there exists a controller that keeps any trajectory starting within the set inside the set under all possible disturbances.
\begin{defn}
Given a discrete-time dynamical system:
\begin{equation}\label{eq:mRCI_model}
  {x^ + } = f(x,u,w), \quad x\in\mathbb{R}^n,u\in\mathcal{U},w\in \mathcal{W}
\end{equation}
where $x$, $u$, and $w$ are the state, control input, and disturbance. A set $\mathcal{S}\subseteq\mathbb{R}^{n}$ is \textit{robust control invariant} if
$\forall x\in \mathcal{S},~\forall w\in\mathcal{W},~\exists~ u\in\mathcal{U} \quad  \mathrm{s.t.}\quad x^+=f(x,u,w)\in \mathcal{S}.$

\end{defn}
In addition, we assume $w=[w^m;w^u]$, where $w^m$ and $w^u$ are the measured and unmeasured disturbances, respectively. The control policy can depend on $w^m$, but not on $w^u$. 

The mRCI algorithm assumes a discrete-time linear model:
\begin{equation}\label{eq:mRCI_linear_model}
  x^+=Ax + Bu + Ew.
\end{equation}
The RCI takes a polytopic form $\Poly(P,q)$ with the hyperplane orientation fixed to $P$. A one-step propagation  computes a new polytope $\Poly(P,q^+)$ that contains all possible $x^+$ with $x\in\Poly(P,q)$, and $w\in\mathcal{W}$ under the dynamics in \eqref{eq:mRCI_linear_model}. We assume the control law takes the form  $u = {K_{ff}}w^m + {K_{fb}}x$, but note that once the mRCI is computed, it is enforced by control barrier functions, as reviewed in Section \ref{sec:CBF}. The linear control law here is simply used to show that there exists a control strategy that renders the set robustly control invariant, it does not have to be implemented.

The one-step propagation is formulated as a robust linear programming problem by assuming a linear form of the control law $u = {K_{ff}}d + {K_{fb}}x$. Robust linear programming (with polytopic uncertainty) is then solved via linear programming  using dualization. Thus, it enjoys order constant complexity.

In the power grid case study, with dynamics $\Sigma_i$ described by \eqref{eq:dyn_RCI}, we assume that phase angles of the neighboring nodes and the local generation and uncontrolled load are measured disturbances. The unmeasured disturbance is due to the communication delay between the neighboring nodes. Suppose  node $i$ and $j$ are neigbors, the bound on frequency is $\omega^{\max}$, and the time delay of communication is $\tau$. Then the maximum difference between the actual value of $\theta_j$ and the value used for feedback is $\omega^{\max}\tau$. The bound of the unmeasured disturbance for the $i^{\text{th}}$ node $w_i^u$ is then given as
$\left| {{w_i^u}} \right| \le \left| {{B_i}K_{ff}^i} \right|{\omega ^{\max }}\tau
$
where $B_i$ and $K_{ff}^i$ are the input matrix and feedforward gain of the $i^{\text{th}}$ node.
The following one-step propagation solves for a polytopic set $\Poly(P,q^+)$ that contains all possible $x^+$ with $x\in\Poly(P,q)$ and $w\in\mathcal{W}$:
\begin{equation}\label{eq:one_step}
\resizebox{.85\hsize}{!}{$
\begin{aligned}
  \mathop {\min }\limits_{{K_{ff}},{K_{fb}} ,{q^ + }} \;&{c^\intercal }{q^ + }\;\hfill \\
 \rm{s.t.}~~~ &\forall~ x \in \Poly(P,q),\forall~ w \in \mathcal{W}, \hfill \\
  &P\left( {A x + B\left( {K_{ff}^\intercal w^m + K_{fb}^\intercal x} \right) + E w} \right)\le {q^ + }, \hfill \\
  &K_{ff}^\intercal w^m + K_{fb}^\intercal x \in \mathcal{U}, \hfill \\
\end{aligned}
$}
\end{equation}
which is solvable with robust linear programming \cite{bertsimas2011theory}.

\begin{rem}
  We enforce an additional constraint that for the generator buses, the frequency stays bounded $|\omega_i|\le\omega^{\max}$, which is easily enforced as a constraint on $q^+$.
\end{rem}
Optimization~\eqref{eq:one_step} solves the one-step propagation problem, which is embedded in an iterative algorithm to find a minimal robust control invariant set. The  algorithm is initiated from a small $q$ and iteratively updates $q$ with $q^+$. If $q^+\le q$, then $\mathcal{P}(P,q)$ is robustly control invariant, and the algorithm terminates, as shown in Algorithm \ref{alg:io}.
\begin{algorithm}
    \caption{Robust LP algorithm for mRCI}
    \label{alg:io}
    \begin{algorithmic}[1] 
        \Procedure{RCI-IO}{$\Sigma$,\ $P$,\ $q^0$,\ $\mathcal{W}$,\ $\mathcal{U}$,\ $\epsilon$}
            \State $q\gets q^0$
            \Do
                \State \resizebox{.8\hsize}{!}{$\begin{gathered}
                        {\text{Find }}\left[ {{q^ + } ,{K_{ff}},{K_{fb}}} \right]\;{\text{s}}{\text{.t}}{\text{. }}\hfill \\
                        \forall~ x \in \Poly(P,q),\forall~ w \in \mathcal{W}, K_{ff}w^m+K_{fb}x\in\mathcal{U},\hfill \\
                        {x^ + } \in \Poly(P,q^+ -\epsilon \mathbf{1}_L) \hfill \\
                        \end{gathered}$}
                \State $q \gets q^+$
            \DoWhile {$q^+\le q+\epsilon \mathbf{1}_L$}
            \State \textbf{return} $\left[ {q ,{K_{ff}},{K_{fb}}} \right]\;$
        \EndProcedure
    \end{algorithmic}
\end{algorithm}
\vspace{-0.3cm}
\section{Proof of Theorem \ref{thm:ag}}\label{sec:proof_theorem_ag}

\noindent  By assumption 3 and 4, $p_{ae}^i$ and $p_{af}^i[0]$ exists so that $\phi_{ae}^i$ and $\phi_{af}^i[0]$ are satisfied. By assumption 1, define the following recursion:
\begin{equation}\label{eq:recursion}
\begin{aligned}
  p_g[k] &= \hat{\Lambda}(p_{af}[k]) \\
  p_{af}[k+1]&=\Gamma(p_g[k]).
\end{aligned}
\end{equation}
Then, we can build an infinite sequence of STLs that the network system satisfies from assumption 1, 2, and \eqref{eq:recursion}:
  \begin{equation*}
  \begin{array}{l}
    \bigwedge\limits_{i=1}^{N}{\varphi_{ae}^i(p_{ae}^i)} \wedge \bigwedge\limits_{i=1}^{N}{\varphi_{af}^i(p_{af}^i[0])} \wedge\\
    \left( \bigwedge\limits_{i=1}^{N}{\varphi_{ae}^i(p_{ae}^i)} \wedge \bigwedge\limits_{i=1}^{N}{\varphi_{af}^i(p_{af}^i[0])} \Rightarrow \bigwedge\limits_{i=1}^{N}{\varphi_g^i(p_g^i[0])}\right)\wedge \\
    \left(\bigwedge\limits_{i=1}^{N}{\varphi_g^i(p_g^i[0])} \Rightarrow \bigwedge\limits_{i=1}^{N}{\varphi_{af}^i(p_{af}^i[1])} \right)\wedge\\
    ...
  \end{array}
\end{equation*}
which implies \eqref{eq:res_guarantee}.
\vspace{-0.5cm}
\section{Proof of Theorem \ref{thm:set_invariance}}\label{sec:proof_theorem_set_inv}
  Let $\mathcal{S}_i=\mathscr{F}_i(y_{\mathcal{N}_i}^{\max})$, and define a network assume-guarantee contract with
  \begin{subequations}\label{eq:inv_contract}
    \begin{align}
     \phi_{ae}^i=&(x_i(0)\in\mathcal{S}_i)\wedge  \square\left(d_i\in\mathcal{D}_i\right)\nonumber\\
    &\wedge\square\left(u_i=k_i(x_i,y_{\mathcal{N}_i},d_i)\right), \label{eq:inv_ae}\\
    \phi_{af}^i =& \varphi_{af}^i(T) = \square_{[0,T]}\left|y_{\mathcal{N}_i}\right|\le y_{\mathcal{N}_i}^{\max},\label{eq:inv_af}\\
    \phi_{g}^i =& \varphi_{g}^i(\hat{T}) = \square_{[0,\hat{T}]}{x_i\in\mathcal{S}_i};\label{eq:inv_g}
    \end{align}
  \end{subequations}
where $k_i$ is the feedback law that keeps $x_i$ within $\mathcal{S}_i$. By the definition of an RCI, the existence of $k_i$ is guaranteed. Let $\hat{\Lambda}(T)=T+T_s$, $\Gamma(\hat{T})=\hat{T}$, where $T_s$ is the time step of the discrete dynamics in \eqref{eq:dynamic_equation}.

Among the 4 assumptions of Theorem \ref{thm:ag}, Assumption 1 is satisfied by the definition of an RCI. With \eqref{eq:validity}, Assumption 2 is satisfied with $\Gamma$ defined above. Assumption 3 is satisfied by \eqref{eq:inv_ae} and Assumption 4 is satisfied by setting $T=0$ in \eqref{eq:inv_af}. Then, by Theorem \ref{thm:ag}, the guarantee for the network system is
  $\hat{\phi}_g^i=\bigwedge\limits_{k=0}^{\infty}{\square_{[0,k\cdot T_s]}{x_i\in\mathcal{S}_i}}$, which is simplified to
$
  \forall~ i=1,...,N, \square_{[0,\infty)}{x_i\in\mathcal{S}_i}.
$

\balance
\bibliographystyle{myieeetran}
\bibliography{Grid_bib}

\end{document}